\documentclass{article}

\usepackage[final]{nips_2017}
\usepackage{graphicx} 
\usepackage{amsmath}
\usepackage{amsthm}
\usepackage{natbib}
\usepackage{caption}
\usepackage{subcaption}
\usepackage{algorithm}
\usepackage{algorithmic}

\usepackage{hyperref}

\usepackage{xr-hyper}
\usepackage{hyperref}
\newcommand{\tr}{{\rm tr}}

\newcommand{\prob}[1]{\ensuremath{\mathbb P}\left(#1\right)}

\usepackage{comment}
\RequirePackage{natbib}
\usepackage{amsmath}
\usepackage{amsthm}
\usepackage{amssymb}
\usepackage{graphicx}
\usepackage{float}
\usepackage{hyperref}
\newtheorem{theorem}{Theorem}

\newtheorem{corollary}[theorem]{Corollary}
\newtheorem{lemma}[theorem]{Lemma}
\usepackage{appendix}
\title{Time-dependent spatially varying graphical models, with application to
brain fMRI data analysis}

\author{
  Kristjan~Greenewald \\
  Department of Statistics\\
  Harvard University\\
  \And
  Seyoung~Park\\
  Department of Biostatistics\\
  Yale University\\
  \And
  Shuheng~Zhou \\
  Department of Statistics\\
  University of Michigan\\
  \AND
  Alexander~Giessing \\
  Department of Statistics\\
  University of Michigan\\
}
\begin{document}
\maketitle
%
%
%
%
%
%




\begin{abstract}
%

In this work, we present an additive model for space-time data that splits
the data into a temporally correlated component and a spatially correlated
component. We model the spatially correlated portion using a time-varying
Gaussian graphical model.
Under assumptions on the smoothness of changes in covariance matrices,  we
derive strong single sample convergence results, confirming our ability to
estimate meaningful graphical structures as they evolve over time. We apply
our methodology to the discovery of time-varying spatial structures in
human brain fMRI signals.
\end{abstract}

\section{Introduction}

Learning structured models of high-dimensional datasets from relatively few training samples is an important task in statistics and machine learning. Spatiotemporal data, in the form of $n$ variables evolving over $m$ time points, often fits this regime due to the high ($mn$) dimension and potential difficulty in obtaining independent samples. 
In this work, we
develop a nonparametric framework for estimating time varying spatiotemporal graphical structure using an $\ell1$ regularization method.
The covariance of a spatiotemporal array $X = [x^1,\dots, x^m] \in \mathbb{R}^{n\times m}$ is an $mn$ by $mn$ matrix
\begin{equation}
{\Sigma}= \mathrm{Cov}\left[\mathrm{vec}([x^1,\dots, x^m]) \right],
\end{equation}
where $x^i \in \mathbb{R}^{n}$, $i = 1,\dots, m$ denotes the $n$ variables or features of interest at the $i$th time point. 
Even for moderately large $m$ and $n$ the number of degrees of freedom ($mn(mn+1)/2$) in the covariance matrix can greatly exceed the number of training samples available for estimation. One way to handle this problem is to introduce structure and/or sparsity, thus reducing the number of parameters to be estimated. Spatiotemporal data is often highly structured, hence the design of estimators that model and exploit appropriate covariance structure can provide significant gains.


We aim to develop a nonparametric framework for estimating time
varying graphical structure for matrix-variate distributions. 
Associated with each $x^i \in \mathbb{R}^{n}$ is its undirected graph $G(i)$.
Under the assumption that the law ${\cal L}(x^i)$ of $x^i$ changes
smoothly, \cite{zhou:TV} introduced a nonparametric method to estimate
the graph sequence $G(1), G(2),\ldots$ assuming that the $x^i \sim \mathcal{N}(0,B(i/m))$ are independent, where $B(t)$ is a smooth function over $t \in [0,1]$ and we have mapped the indices $i$ onto points $t = i/m$ on the interval $[0,1]$.
In this work, we are interested in the general time series model where the 
$x^i, i=1, \ldots, m$ are dependent and the $B^{-1}(t)$ graphs change over time.  

One way to introduce dependency into the $x^i$ is to study the following covariance structure. 
Let $A = (a_{ij}) \in \mathbb{R}^{m \times m}, B(t) = (b_{ij}(t) )\in \mathbb{R}^{n \times n}, t \in [0,1]$ be symmetric positive definite
covariance matrices. Let $\mathrm{diag}(v)$, $v = (v_1, \ldots, v_m)$ be the diagonal matrix with elements $v_i$ along the
diagonal. Consider the random matrix $X$ with row vectors $y^j$ corresponding to measurements at the $j$th spatial location, and columns $x^i$ corresponding to the $m$ measurements at times $i/m$, $i = 1,\dots, m$:
\begin{align}
\label{eq::time}
 & \forall j =1, \ldots, n, \; \; y^j \sim \mathcal{N}_m(0, A_j) \;
 \;\text{where} \; \;
  A_j =  A +\mathrm{diag}(b_{jj}(1), \ldots, b_{jj}(m)), \; \text{ and }\; \\
\label{eq::space}
&  \forall i=1, \ldots, m, \; \; 
x^i \sim \mathcal{N}_n(0, a_{ii} I+ B(i/m)) \; \; \text{ where $B(t)$ changes smoothly over $t \in [0,1]$;}
\end{align}
that is, the covariance of the column vectors $x^i$ corresponding to each time
point changes smoothly with time (if
$a_{ii}$ is a smooth function of $i$). This provides ultimate flexibility in parameterizing spatial
correlations, for example, across different geographical scales
through variograms~\citep{cressie2015statistics}, each of which is allowed to change over seasons. Observe that while we have used the normal distribution here for simplicity, all our results hold for general subgaussian distributions.

The model \eqref{eq::space} also allows modeling the dynamic gene regulatory and brain connectivity networks with topological (e.g.,
Erd\H{o}s-R\'{e}nyi random graph, small-world graph, or modular
graphs) constraints via degree specifications as well as spatial
constraints in the set of $\{B(t), t=1, 2,\ldots\}$. 
When $A = {\bf 0}$, we return to the
case of~\cite{zhou:TV} where there is no temporal correlation, i.e. $y^1, \ldots, y^n$ assumed to be independent. 

We propose methodologies to study the model as constructed
in~\eqref{eq::time} and~\eqref{eq::space}. Building upon and extending techniques of~\cite{zhou:TV}
and~\cite{rudelson2015high,zhou2014gemini}, we aim to design estimators to estimate graph sequence $G(1),
G(2),\ldots$, where the temporal graph $H$ and spatial graphs
$G(i)$ are determined by the zeros of $A^{-1}$ and $B(t)^{-1}$.   
Intuitively, the temporal correlation and spatial correlation are
modeled as two additive processes. The covariance of $X$ is now
\begin{equation}
\label{eq:kronCov}
\mathrm{Cov} [\mathrm{vec}(X)] = \Sigma = A \otimes I_n + \sum\nolimits_{i = 1}^m (e_i e_i^T) \otimes B(i/m)
\end{equation}
where $e_i \in \mathbb{R}^m, \ \forall i$ are the $m$-dimensional standard basis vectors.


In the context of this model, we aim to develop a nonparametric method for estimating time varying graphical structure for matrix-variate normal distributions using an $\ell_1$ regularization method. We will show that, as long as the covariances change smoothly over time, we can estimate the spatial and temporal covariance matrices well in terms of predictive risk even when $n, m$ are both large. We will investigate the following theoretical properties: (a) consistency and rate of convergence in the operator and Frobenius norm of the covariance matrices and their inverse, (b) large deviation results for covariance matrices for simultaneously correlated and non-identically distributed observations, and (c) conditions that guarantee smoothness of the covariances.

Besides the model \eqref{eq:kronCov}, another well-studied option for modeling spatio-temporal covariances $\Sigma$ is to introduce structure via the Kronecker product of smaller symmetric positive definite matrices, i.e.
$
{\Sigma} = {A}\otimes {B}.
$
The Kronecker product model, however, is restrictive when applied to general spatio-temporal covariances as it assumes the covariance is separable (disallowing such simple scenarios as the presence of additive noise), and does not allow for time varying spatial structure. When used to estimate covariances not following Kronecker product structure, many estimators will respond to the model mismatch by giving ill-conditioned estimates \citep{greenewaldTSP}.

Human neuroscience data is a notable application where time-varying structure emerges. In neuroscience, one must take into account temporal correlations as
well as spatial correlations, which reflect the connectivity formed by the
neural pathways. It is conceivable that the brain
connectivity graph will change over a sufficiently long period of
measurements. For example, as a child learns to associate symbols in the
environment, certain pathways within the brain are reinforced. When they begin to associate images with words, the
correlation between a particular sound like Mommy and the sight of a face
becomes stronger and forms a well worn pathway. On the other hand, long
term non-use of connections between sensory and motor neurons can
result in a loss of the pathway. 

\subsection{Datasets and Related Work}


Estimating graphical models (connectomes) in fMRI data using sparse inverse covariance techniques has enjoyed wide application \citep{huang2010learning,varoquaux2010, narayan2015two, kim2015highly}. However, recent research has only now begun exploring observed phenomena such as temporal correlations and additive temporally correlated noise \citep{chen2015m,arbabshirani2014impact,kim2015highly, qiu2016joint}, and time-varying dynamics and graphical models (connectomes) \citep{calhoun2014chronnectome,liu2013time,chang2010time,chen2015m}. 


We consider the ADHD-200 fMRI dataset \citep{biswal2010toward}, and study resting state fMRIs for a variety of healthy patients in the dataset at different stages of development. Using our methods, we are able to directly estimate age-varying graphical models across brain regions, chronicling the development of brain structure throughout childhood.

Several models have emerged to generalize the Kronecker product model to allow it to model more realistic covariances while still maintaining many of the gains associated with Kronecker structure. Kronecker PCA, discussed in \cite{tsiliArxiv}, approximates the covariance matrix using a sum of Kronecker products.
An algorithm (Permuted Rank-penalized Least Squares (PRLS)) for fitting the KronPCA model to a measured sample covariance matrix was introduced in \citep{tsiliArxiv} and was shown to have strong high dimensional MSE performance guarantees. 
From a modeling perspective, the strengths of Kronecker PCA lie in its ability to handle ``near separable" covariances and a variety of time-varying effects. While the Kronecker PCA model is very general, so far incorporation of sparsity in the inverse covariance has not been forthcoming. This motivates our introduction of the sparse model \eqref{eq:kronCov}, which we demonstrate experimentally in Section \ref{supp:compare} of the supplement to enjoy better statistical convergence.

\cite{carvalho2007dynamic} proposed a Bayesian additive time-varying graphical model, where the spatially-correlated noise term is a parameter of the driving noise covariance in a temporal dynamical model. Unlike our method, they did not estimate the temporal correlation, instead requiring the dynamical model to be pre-set. Our proposed method has wholly independent spatial and temporal models, directly estimating an inverse covariance graphical model for the temporal relationships of the data. This allows for a much richer temporal model and increases its applicability.

In the context of fMRI, the work of \cite{qiu2016joint} used a similar kernel-weighted estimator for the spatial covariance, however they modeled the temporal covariance with a simple AR-1 model which they did not estimate, and their estimator did not attempt to remove. Similarly, \cite{monti2014estimating} used a smoothed kernel estimator for $B^{-1}(t)$ with a penalty to further promote smoothness, but did not model the temporal correlations. Our additive model allows the direct estimation of the temporal behavior, revealing a richer structure than a simple AR-1, and allowing for effective denoising of the data, and hence better
estimation of the spatial graph structures. 





\section{The model and method}\label{sec:ModelFramework}



Let the elements of $A \succ 0$ and $B(t)$ be denoted as $[A]_{ij} := a_{ij}$ and $[B(t)]_{ij} := b_{ij}(t),  t \in [0,1]$.  
Similar to the setting in \citep{zhou:TV}, we assume that $b_{ij}(t)$ is a smooth function of time $t$ for all $i,j$, and assume that $B^{-1}(t)$ is sparse. Furthermore, we suppose that $m \gg n$, corresponding to there being more time points than spatial variables.
 For a random variable $Y$, the subgaussian norm of $Y$, $\|Y\|_{\psi_2}$, is defined via 
$
 \|Y\|_{\psi_2} = \sup_{p \geq 1} p^{-1/2}(E|Y|^p)^{1/p}$.
 Note that if $E[Y] = 0$, we also have $E[\exp(tY)] \leq \exp(Ct^2\|Y\|_{\psi_2}^2) \quad \forall t \in \mathbb{R}$.
Define an $n \times m$ random matrix $Z$ with independent, zero mean entries $Z_{ij}$ satisfying $E[Z_{ij}^2] = 1$ and having subgaussian norm $\|Z_{ij}\|_{\psi_2} \leq K$. Matrices $Z_1, Z_2$ denote independent copies of $Z$. 
We now write an additive generative model for subgaussian data $X \in \mathbb{R}^{n \times m}$ having covariance given in \eqref{eq:kronCov}. Let
\begin{equation}
\label{eq:TVMVGM}
X = Z_1 A^{1/2} + Z_B
\end{equation}
where
$Z_B = [B(1/m)^{1/2} Z_2 e_1, \dots, B(i/m)^{1/2} Z_2 e_i, \dots, B(1)^{1/2} Z_2 e_m],
$
and $e_i \in \mathbb{R}^m, \ \forall i$ are the $m$-dimensional standard basis vectors. Then the covariance
\begin{align*}
&\Sigma = \mathrm{Cov}[\mathrm{vec}(X)] = \mathrm{Cov}[\mathrm{vec}(Z_1 A^{1/2})] + \mathrm{Cov}[\mathrm{vec}(Z_B)]\\&= \mathrm{Cov}[\mathrm{vec}(Z_1 A^{1/2})] + \sum\nolimits_{i=1}^m  (e_i e_i^T) \otimes \mathrm{Cov}[B(i/m)^{1/2} Z_2 e_i]\\ & = A\otimes I_n + \sum\nolimits_{i=1}^m (e_i e_i^T) \otimes B(i/m).
\end{align*}
Thus \eqref{eq:TVMVGM} is a generative model for data following the covariance model \eqref{eq:kronCov}.


\subsection{Estimators}

As in \cite{rudelson2015high}, we can exploit the large-$m$ convergence of $Z_1 A Z_1^T$ to $\tr(A) I$ to project out the $A$ part and create an estimator for the $B$ covariances. As $B(t)$ is time-varying, we use a weighted average across time to create local estimators of spatial covariance matrix $B(t)$. 

It is often assumed that knowledge of the trace of one of the factors
is available a priori.
For example, the spatial signal variance may be known and time invariant, corresponding to $\tr(B(t))$ being known. Alternatively, the temporal component variance may be constant and known, corresponding to $\tr(A)$ being known. In our analysis below, we suppose that $\tr (A)$ is known or otherwise estimated (similar results hold when $\tr(B(t))$ is known). For simplicity in stating the trace estimators, in what follows we suppose that $\tr(B(t)) = \tr(B)$ is constant, and without loss of generality that the data has been normalized such that diagonal elements $A_{ii}$ are constant over $i$. 

 As $B(t)$ is smoothly varying over time, the estimate at time $t_0$ should depend strongly on the time samples close to $t_0$, and less on the samples further from $t_0$. For any time of interest $t_0$, we thus construct a weighted estimator using a weight vector $w_i (t_0)$ such that $\sum_{t = 1}^m  w_t(t_0) = 1$. 
Our weighted, unstructured sample-based estimator for $B(t_0)$ is then given by 
\begin{align}\label{eq:PlugInEst}
\widehat{S}_m(t_0) := \sum\nolimits_{i = 1}^m w_{i}({t_0}) \left(x_i  x_i^T - \frac{\tr(A)}{m} I_n\right), \: \mathrm{where} \quad w_{i}(t_0) =  \frac{1}{mh} K\left(\frac{i/m-t_0}{h}\right),
\end{align}
and we have considered the class of weight vectors $w_i(t_0)$
arising from a symmetric nonnegative kernel function $K$ with compact support $[0,1]$ and bandwidth determined by parameter $h$. A list of minor regularity assumptions on $K$ are listed in the supplement.
For kernels such as the Gaussian kernel, this $w_t(t_0)$ will result in samples close to $t_0$ being highly weighted, with the ``weight decay" away from $t_0$ scaling with the bandwidth $h$. A wide bandwidth will be appropriate for slowly-varying $B(t)$, and a narrow bandwidth for quickly-varying $B(t)$.

To enforce sparsity in the estimator for $B^{-1}(t_0)$, we substitute $\widehat{S}_m(t_0)$ into the widely-used GLasso objective function, resulting in a penalized estimator for $B(t_0)$ with regularization parameter $\lambda_m$ 
\begin{equation}
\label{eq:ColEst}
\widehat{B}_\lambda(t_0) := \arg\min_{B_\lambda\succ 0} \mathrm{tr}\left(B_\lambda^{-1} \widehat{S}_m(t_0)\right) + \log |B_\lambda| + \lambda_m |B_\lambda^{-1}|_{1}.
\end{equation}
For a matrix $B$, we let $|B|_{1} := \sum_{i  j} |B_{ij}|$. Increasing the parameter $\lambda_m$ gives an increasingly sparse $\widehat{B}_\lambda^{-1}(t_0)$.
Having formed an estimator for $B$, we can now form a similar estimator for $A$. Under the constant-trace assumption, we construct an estimator for $\tr(B)$
\begin{align}
\hat{\tr}(B) &= \sum\nolimits_{i = 1}^m w_i\|X_i \|_2^2 - \frac{n}{m} \tr(A), \mathrm{with}\: w_i = \frac{1}{m}.
\end{align}
For a time-varying trace $\tr(B(t))$, use the time-averaged kernel 
\begin{align}
\hat{\tr}(B(t_0)) &= \sum_{i = 1}^m w_i(t_0) \|X_i \|_2^2 - \frac{n}{m} \tr(A), \mathrm{with}\: w_i(t_0) = \frac{1}{mh}  K \left(\frac{ i/m  -t_0}{h}\right).
\end{align}
In the future we will derive rates for the time varying case by choosing an optimal $h$.
These estimators allow us to construct a sample covariance matrix for $A$:
\begin{align} 
\label{eq:atilde}
\tilde{A} &= \frac{1}{n} X^T X - \frac{1}{n}\mathrm{diag}\{\hat{\mathrm{tr}}(B(1/m)), \dots,\hat{\mathrm{tr}}(B(1))\}. 
\end{align}
We (similarly to $B(t)$) apply the GLasso approach to $\tilde{A}$. Note that with $m > n$, $\tilde{A}$ has negative eigenvalues since $\lambda_{\min}\left(\frac{1}{n} X^T X\right)=0$.
We obtain a positive semidefinite matrix $\tilde{A}_+$ as: 
\begin{equation}
\label{solve:A}
\tilde{A}_{+} = \arg\min_{A \succeq 0} \|\tilde{A}-A\|_{\max}.
\end{equation}
We use alternating direction method of multipliers (ADMM) to solve \eqref{solve:A} as in \citet{Boyd:2011}, and prove that this retains a tight elementwise error bound. Note that while we chose this method of obtaining a positive semidefinite $\tilde{A}_+$ for its simplicity, there may exist other possible projections, the exact method is not critical to our overall Kronecker sum approach. In fact, if the GLasso is not used, it is not necessary to do the projection \eqref{solve:A}, as the elementwise bounds also hold for $\tilde{A}$.

We provide a regularized estimator for the correlation matrices $\rho(A)= \mathrm{diag}(A)^{-1/2} A \mathrm{diag}(A)^{-1/2}$ using the positive semidefinite $\tilde{A}_+$ as the initial input to the GLasso problem
\begin{equation}
\label{eq:RowEst}
\hat{\rho}_{\lambda}(A)= \mathrm{argmin}_{A_\rho \succ 0}\  \mathrm{tr}(A_\rho^{-1} \rho(\tilde{A}_+)) + \log|A_\rho| + \lambda_n |A^{-1}_\rho|_{1,\mathrm{off}},
\end{equation}
where $\lambda_n > 0$ is a regularization parameter and $|\cdot|_{1,\mathrm{off}}$ is the L1 norm on the offdiagonal. 

Form the estimate for $A$ as $\frac{\tr(A)}{m} \hat{\rho}_{\lambda}(A)$. Observe that our method has three tuning parameters, two if $\tr(A)$ is known or can be estimated. If $\tr(A)$ is not known, we present several methods to choose it in Section \ref{supp:tuning} in the supplement.  Once $\tr(A)$ is chosen, the estimators \eqref{eq:ColEst} and \eqref{eq:RowEst} for $A$ and $B(t)$ respectively do not depend on each other, allowing $\lambda_m$ and $\lambda_n$ to be tuned independently.

\section{Statistical convergence}

We first bound the estimation error for the time-varying $B(t)$. Since $\hat{B}(t)$ is based on a kernel-smoothed sample covariance, $\hat{B}(t)$ is a biased estimator, with the bias depending on the kernel width and the smoothness of $B(t)$. In Section \ref{supp:prel} of the supplement, we derive the bias and variance of $\hat{S}_m(t_0)$, using arguments from kernel smoothing and subgaussian concentration respectively.


In the following results, we assume that the eigenvalues of the matrices $A$ and $B(t)$ are bounded:

\textbf{Assumption 1}: There exist positive constants $c_A, c_B$ such that $\frac{1}{c_A} \leq \lambda_{\min}(A) \leq \lambda_{\max}(A) \leq c_A$ and $\frac{1}{c_B} \leq \lambda_{\min}(B(t)) \leq \lambda_{\max}(B(t)) \leq c_B$ for all $t$.

\textbf{Assumption 2}: $B(t)$ has entries with bounded second derivatives on $[0,1]$.

Putting the bounds on the bias and variance together and optimizing the rate of $h$, we obtain the following, which we prove in the supplementary material.
\begin{theorem}
\label{thm}


Suppose that the above Assumption holds, the entries $B_{ij}(t)$ of  $B(t)$ have bounded second derivatives for all $i,j$, and $t \in [0,1]$, $s_b+n=o((m/\log m)^{2/3})$, and that $h \asymp (m^{-1}\log m)^{1/3}$. Then with probability at least $1 - \frac{c''}{m^{8/3}}$, $\widehat{S}_m(t_0)$ is positive definite and for some $C$
\[
\mathrm{max}_{ij} |\widehat{S}_m(t_0, i, j) - B({t_0, i, j})| \leq C \left(m^{-1}\log m\right)^{1/3}.
\]
\end{theorem}
This result confirms that the $mh$ temporal samples selected by the kernel act as replicates for estimating $B(t)$.
We can now substitute this elementwise bound on $\widehat{S}_m(t_0)$ into the GLasso proof, obtaining the following theorem which demonstrates that $\hat{B}(t)$ successfully exploits sparsity in $B^{-1}(t)$.
\begin{theorem}
\label{thm:GLasso}

Suppose the conditions of Theorem \ref{thm} and that $B^{-1}(t)$ has at most $s_b$ nonzero off-diagonal elements for all $t$.  If $\lambda_m \sim \sqrt{\frac{\log m}{m^{2/3}}}$, then the GLasso estimator \eqref{eq:ColEst} satisfies
\begin{align*}
\|\hat{B}(t_0) - B(t_0) \|_F = O_p \left(\sqrt{\frac{(s_b+n)\log m}{m^{2/3}}}\right), \|\hat{B}^{-1}(t_0) - B^{-1}(t_0) \|_F &= O \left(\sqrt{\frac{(s_b+n)\log m}{m^{2/3}}}\right)
\end{align*}

\end{theorem}
Observe that this single-sample bound converges whenever the $A$ part dimensionality $m$ grows. 
The proof follows from the concentration bound in Theorem \ref{thm} using the argument in \cite{zhou:TV}, \cite{zhou2011high}, and \cite{rothman2008sparse}. Note that $\lambda_m$ goes to zero as $m$ increases, in accordance with the standard bias/variance tradeoff. 

We now turn to the estimator for the $A$ part. As it does not involve kernel smoothing, we simply need to bound the variance. We have the following bound on the error of $\tilde{A}$: 
\begin{theorem}
\label{lem:Apart2}
Suppose the above Assumption holds. 
Then 
\[
\mathrm{max}_{ij} |\tilde{A}_{ij} - A_{ij}| \le C (c_A + c_B) \sqrt{n^{-1}\log m}
\]
with probability $1-\frac{c}{m^4}$ for some constants $C,c>0$.
\end{theorem}
Recall that we have assumed that $m > n$, so the probability converges to 1 with increasing $m$ or $n$.
While $\tilde{A}$ is not positive definite, the triangle inequality implies a bound on the positive definite projection $\tilde{A}_+$ \eqref{solve:A}:
\begin{align}
\label{proj:pos}
\|\tilde{A}_+ - A\|_{\max} &\le \|\tilde{A}_+ - \tilde{A}\|_{\max} + \|\tilde{A}-A\|_{\max} \le 2\|\tilde{A}-A\|_{\max} = O_p\left(\sqrt{n^{-1}{\log m}} \right).
\end{align}
Thus, similarly to the earlier result for $B(t)$, the estimator \eqref{eq:RowEst} formed by substituting the positive semidefinite $\rho(\tilde{A}_+)$ into the GLasso objective enjoys the following error bound \citep{zhou2011high}.
\begin{theorem}
\label{thm:GLasso}

Suppose the conditions of Theorem \ref{lem:Apart2} and that $A^{-1}$ has at most $s_a =o(n/\log m)$ nonzero off-diagonal elements.  If $\lambda_n \sim \sqrt{\frac{\log m}{n}}$, then the GLasso estimator \eqref{eq:RowEst} satisfies
\begin{align*}
\|\hat{A} -A \|_F &= O_p \left(\sqrt{\frac{s_a\log m}{n}}\right),\quad \|\hat{A}^{-1} - A^{-1} \|_F = O_p \left(\sqrt{\frac{s_a\log m}{n}}\right).
\end{align*}

\end{theorem}
Observe that this single-sample bound converges whenever the $B(t)$
dimensionality
$n$ grows since the sparsity $s_a = o(n/\log m)$.  For relaxation of
this stringent sparsity assumption, one can use other assumptions, see for
example Theorem 3.3 in \cite{zhou2014gemini}.

\section{Simulation study}
\label{sec:sims}


%
%
%
%
%
%


We generated a time varying sequence of spatial covariances $B(t_i) = \Theta(t_i)^{-1}$ according to the method of \cite{zhou:TV}, which follows a type of Erdos-Renyi random graph model. Initially we set $\Theta(0) = 0.25I_{n \times n}$, where $n = 100$. Then, we randomly select $k$ edges and update $\Theta(t)$ as follows: for each new edge $(i,j)$, a weight $a>0$ is chosen uniformly at random from $[0.1,0.3]$; we subtract $a$ from $\Theta_{ij}$ and $\Theta_{ji}$, and increase $\Theta_{ii},\Theta_{jj}$ by $a$. This keeps $B(t)$ positive definite. When we later delete an existing edge from the graph, we reverse the above procedure. 

We consider $t \in [0,1]$, changing the graph structure at the points $t_i = i/5$ as follows. At each $t_i$, five existing edges are deleted, and five new edges are added. For each of the five new edges, a target weight is chosen. Linear interpolation of the edge weights between the $t_i$ is used to smoothly add the new edges and gradually delete the ones to be deleted. Thus, almost always, there are 105 edges in the graph and 10 edges have weights that are varying smoothly (Figure \ref{Fig:Graphs}). 

\begin{figure}[h]
\centering
\includegraphics[width=3.4in]{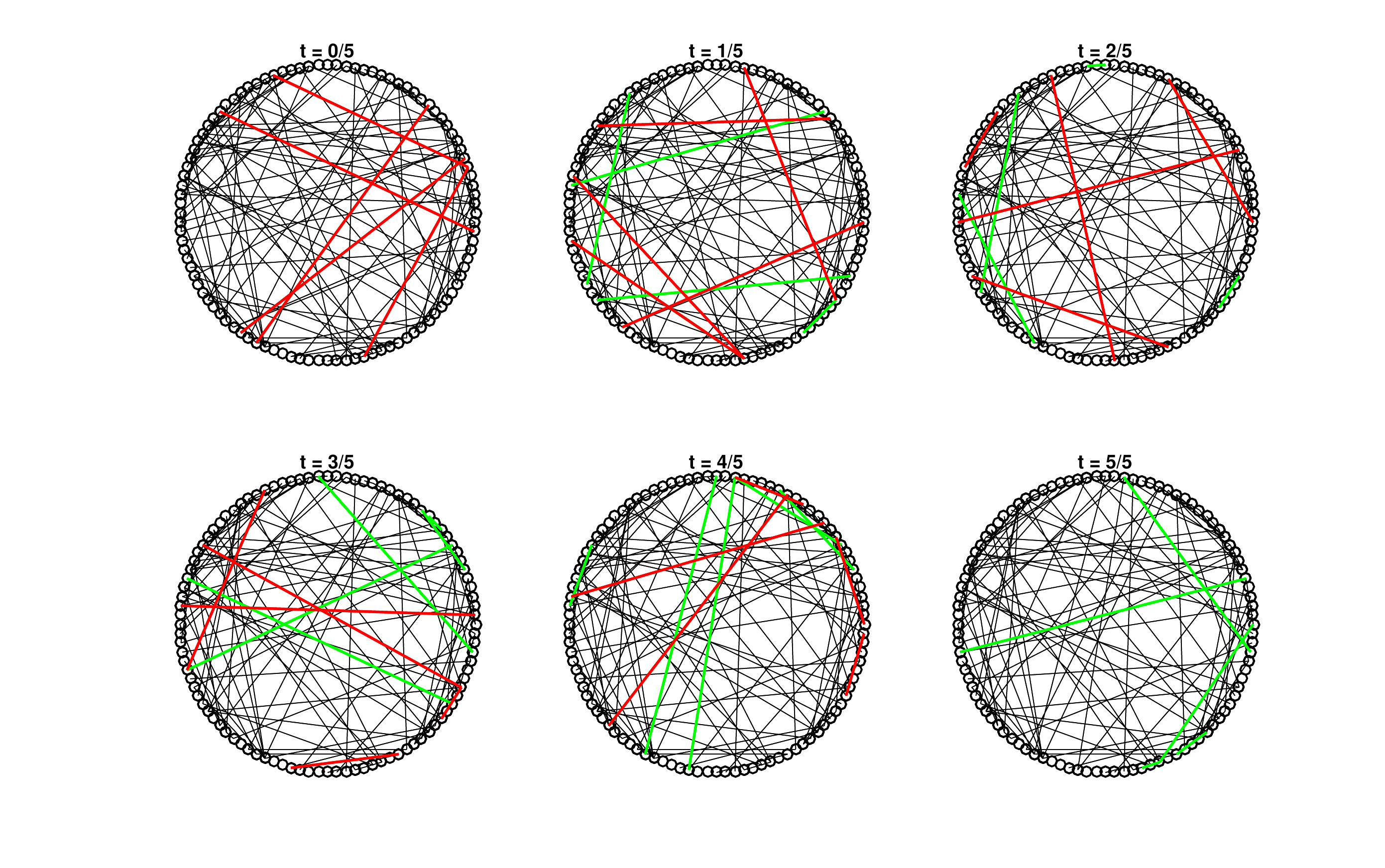}
\caption{Example sequence of Erdos-Renyi $B^{-1}(t) = \Theta(t)$ graphs. At each time point, the 100 edges connecting $n=100$ nodes are shown. Changes are indicated by red and green edges: red edges indicate edges that will be deleted in the next increment and green indicates new edges. }
\label{Fig:Graphs}
\end{figure}


In the first set of experiments we consider $B(t)$ generated from the ER time-varying graph procedure and $A$ an AR-1 covariance with parameter $\rho$. The magnitudes of the two factors are balanced. We set $n = 100$ and vary $m$ from 200 to 2400. For each $n,m$ pair, we vary the $B(t)$ regularization parameter $\lambda$, estimating every $B(t)$, $t = 1/m,\dots, 1$ for each. We evaluate performance using the mean relative Frobenius $B(t)$ estimation error ($\|\hat{B}(t) - B(t)\|_F/\|B(t)\|_F$), the mean relative L2 estimation error ($\|\hat{B}(t) - B(t)\|_2/\|B(t)\|_2$), and the Matthews correlation coefficient (MCC). 

The MCC quantifies edge support estimation performance, and is defined as follows. Let the number of true positive edge detections be TP, true negatives TN, false positives FP, and false negatives FN. The Matthews correlation coefficient is defined as 
$
\mathrm{MCC} = \frac{\mathrm{TP \cdot TN}-\mathrm{FP \cdot FN}}{\sqrt{\mathrm{(TP + FP)(TP + FN)(TN+FP)(TN + FN)}}}.
$
Increasing values of MCC imply better edge estimation performance, with $\mathrm{MCC} = 0$ implying complete failure and $\mathrm{MCC} = 1$ implying perfect edge set estimation.

Results are shown in Figure \ref{Fig:BER}, for $\rho = .5$ and 50 edges in $B$, $\rho = .5$ and 100 edges in $B$, and $\rho = .95$ and 100 edges in $B$. As predicted by the theory, increasing $m$ improves performance and increasing $\rho$ decreases performance. Increasing the number of edges in $B$ changes the optimal $\lambda$, as expected. 
Figure \ref{Fig:AR1} shows performance results for the penalized estimator $\hat{A}$ using MCC, Frobenius error, and L2 error, where $A$ follows an AR(1) model with $\rho=0.5$ and $B$ follows a random ER model.
Note the MCC, Frobenius, spectral norm errors are improved with larger $n$.
In the supplement (Section \ref{supp:randgrid}), we repeat these experiments, using an alternate random graph topologies, with similar results.




\begin{figure}[htb]
\centering
\includegraphics[width=3.25in]{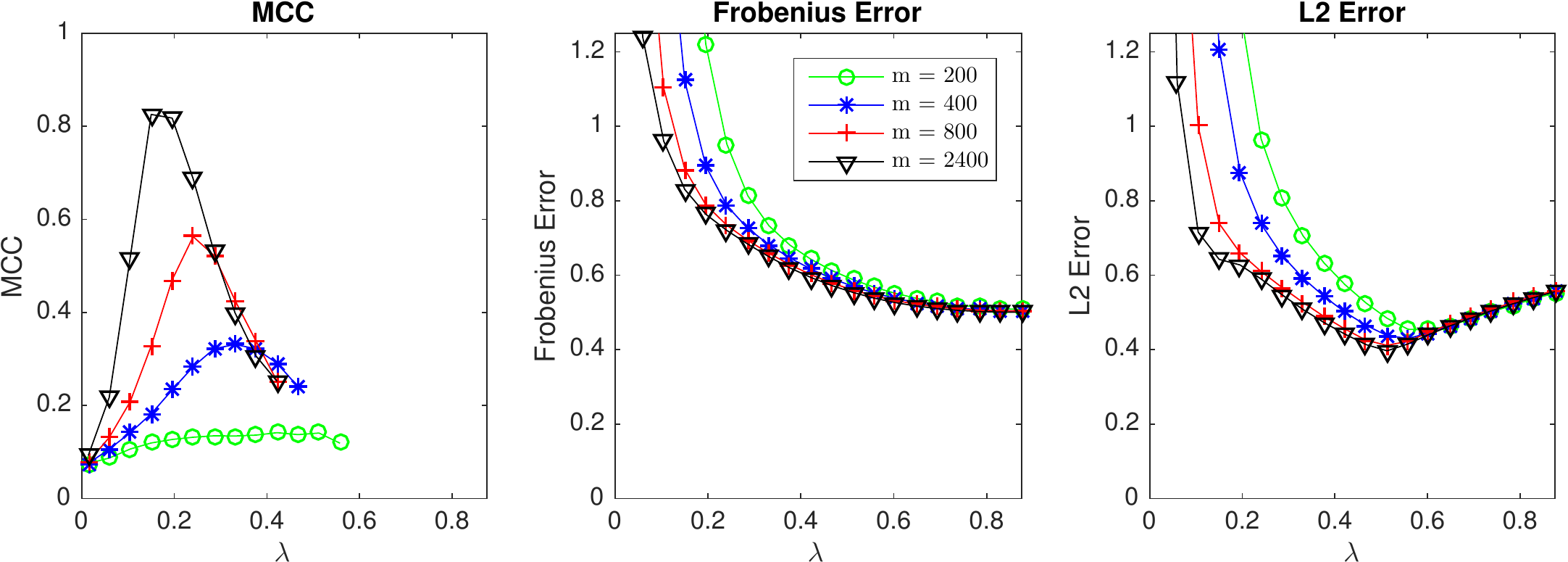}\\
\includegraphics[width=3.25in]{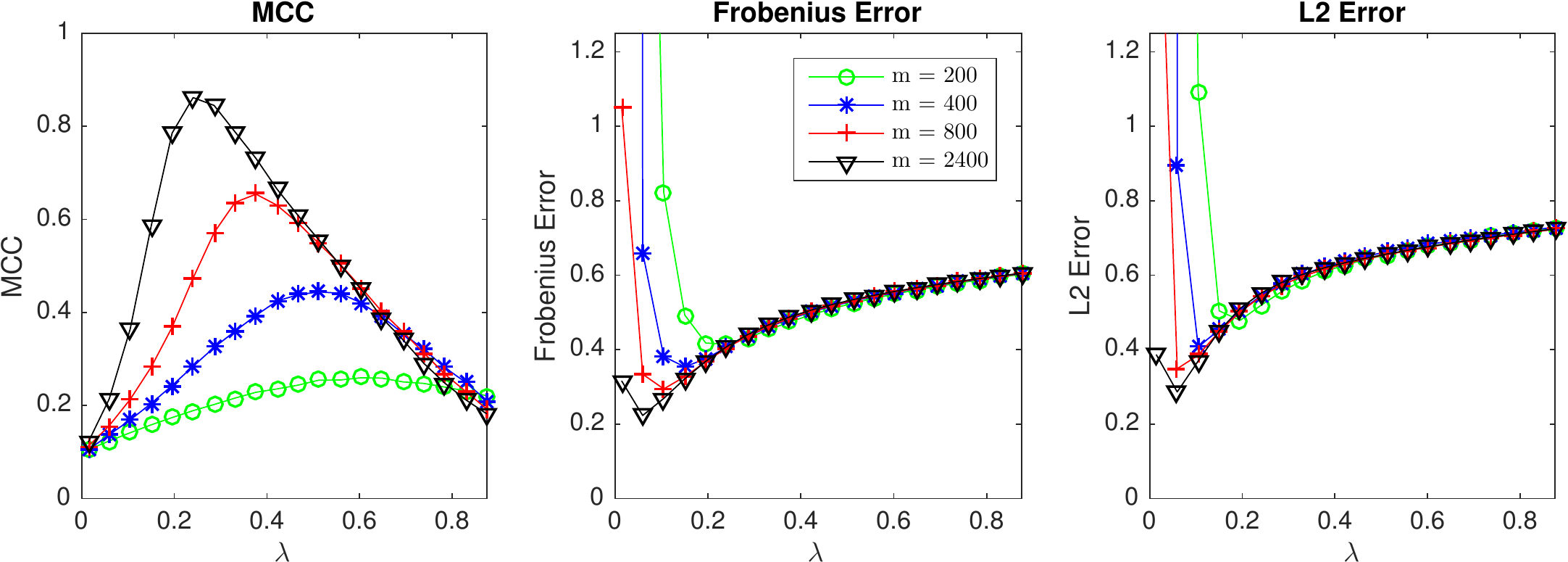}\\
\includegraphics[width=3.25in]{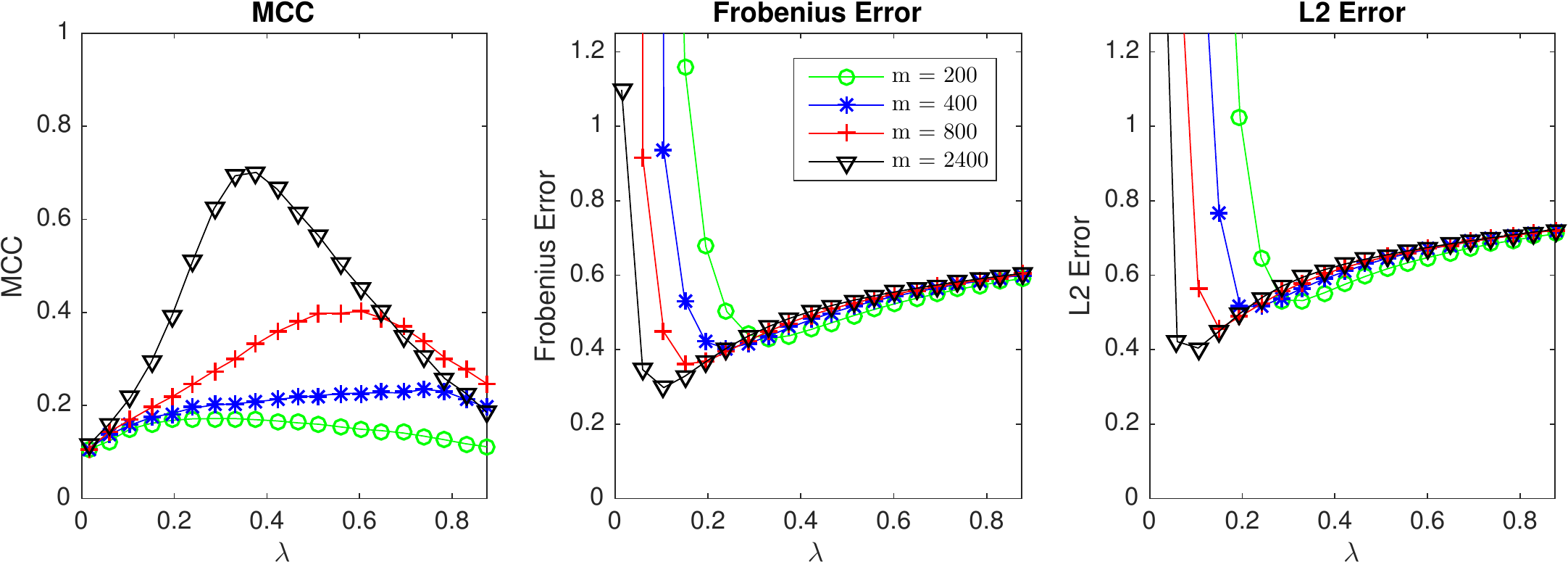}
\caption{MCC, Frobenius, and L2 norm error curves for $B$ a random ER graph and $n = 100$. Top: $A$ is AR covariance with $\rho = .5$ and 50 edges in $B$, Middle: $A$ is AR(1) covariance with $\rho = .5$ and $B$ having 100 edges, Bottom: AR covariance with $\rho = .95$ and 100 edges in $B$. }
\label{Fig:BER}
\end{figure}


\begin{figure}[h]
\centering
\includegraphics[width=3.25in]{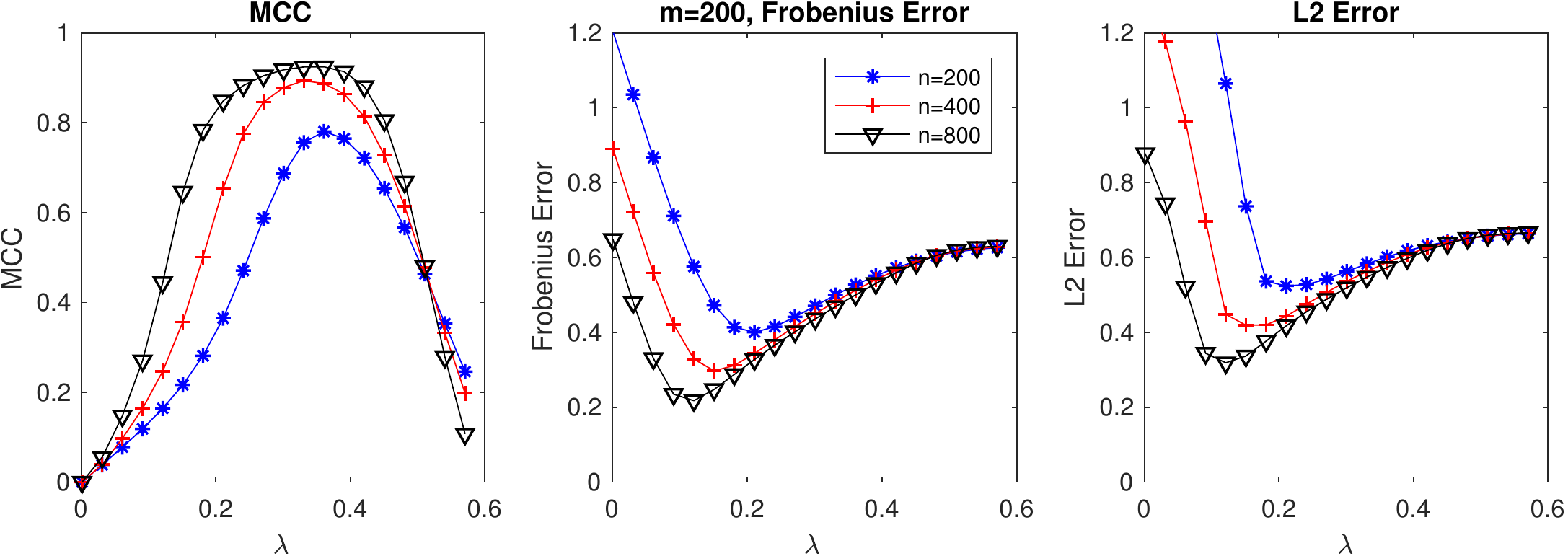}
\\\includegraphics[width=3.25in]{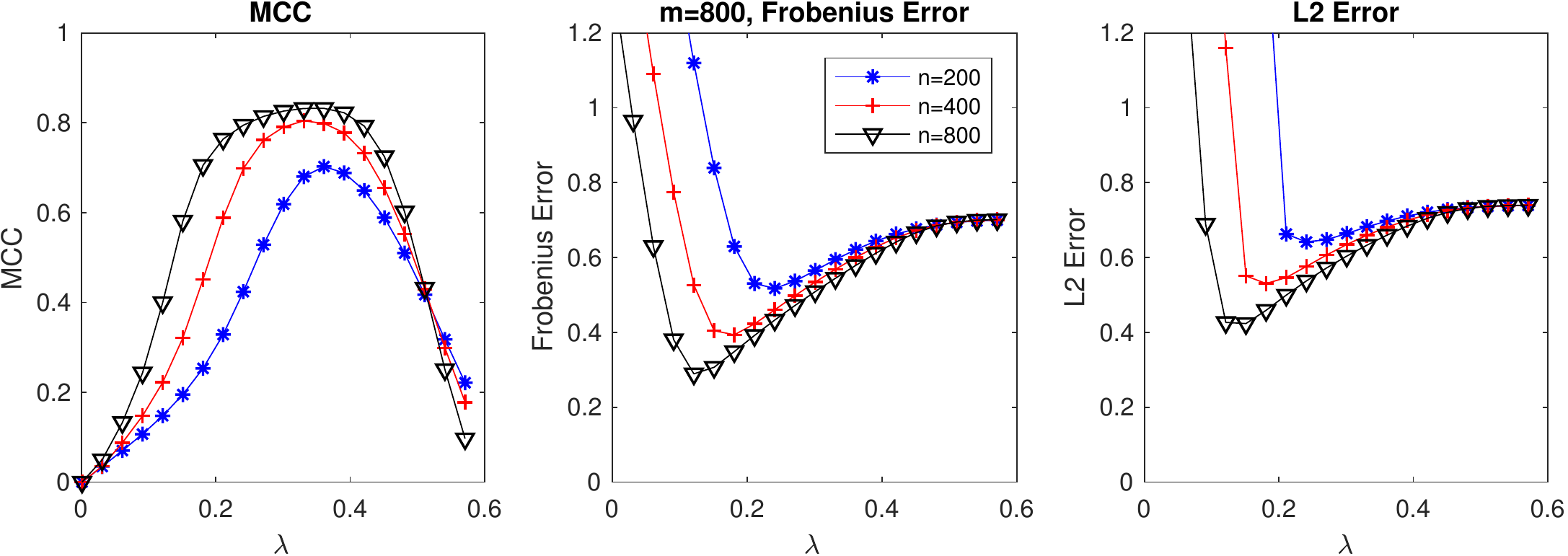}
\caption{MCC, Frobenius, and L2 norm error curves for $A$ a AR(1) with $\rho=0.5$ when $B$ is a random ER graph. From top to bottom: $m=200$ and $m=800$.}
\label{Fig:AR1}
\end{figure}

\section{fMRI Application}
The ADHD-200 resting-state fMRI dataset \citep{biswal2010toward} was collected from 973 subjects, 197 of which were diagnosed with ADHD types 1, 2, or 3. The fMRI images have varying numbers of voxels which we divide into 90 regions of interest for graphical model analysis \citep{wehbe2014simultaneously}, and between 76 and 276 images exist for each subject. Provided covariates for the subjects include age, gender, handedness, and IQ. 
Previous works such as \citep{qiu2016joint} used this dataset to establish that the brain network density increases with age, corresponding to brain development as subjects mature. We revisit this problem using our additive approach. 
Our additive model allows the direct estimation of the temporal behavior, revealing a richer structure than a simple AR-1, and allowing for effectively a denoising of the data, and better estimation of the spatial graph structure.  




We estimate the temporal $A$ covariances for each subject using the voxels contained in the regions of interest, with example results shown in Figure \ref{Fig:fA} in the supplement. We choose $\tau_B$ as the lower limit of the eigenvalues of $\frac{1}{n}X^T X$, as in the high sample regime it is an upper bound on $\tau_B$. 

We then estimate the brain connectivity network at a range of ages from 8 to 18, using both our proposed method and the method of \cite{monti2014estimating}, as it is an optimally-penalized version of the estimator in \cite{qiu2016joint}. We use a Gaussian kernel with bandwidth $h$, and estimate the graphs using a variety of values of $\lambda$ and $h$. Subjects with fewer than 120 time samples were eliminated, and those with more were truncated to 120 to reduce bias towards longer scans. The number of edges in the estimated graphs are shown in Figure \ref{Fig:fEdges}. Note the consistent increase in network density with age, becoming more smooth with increasing $h$. 


\begin{figure}[h]
\centering
\begin{subfigure}[b]{\textwidth}
\centering
\includegraphics[width=5in]{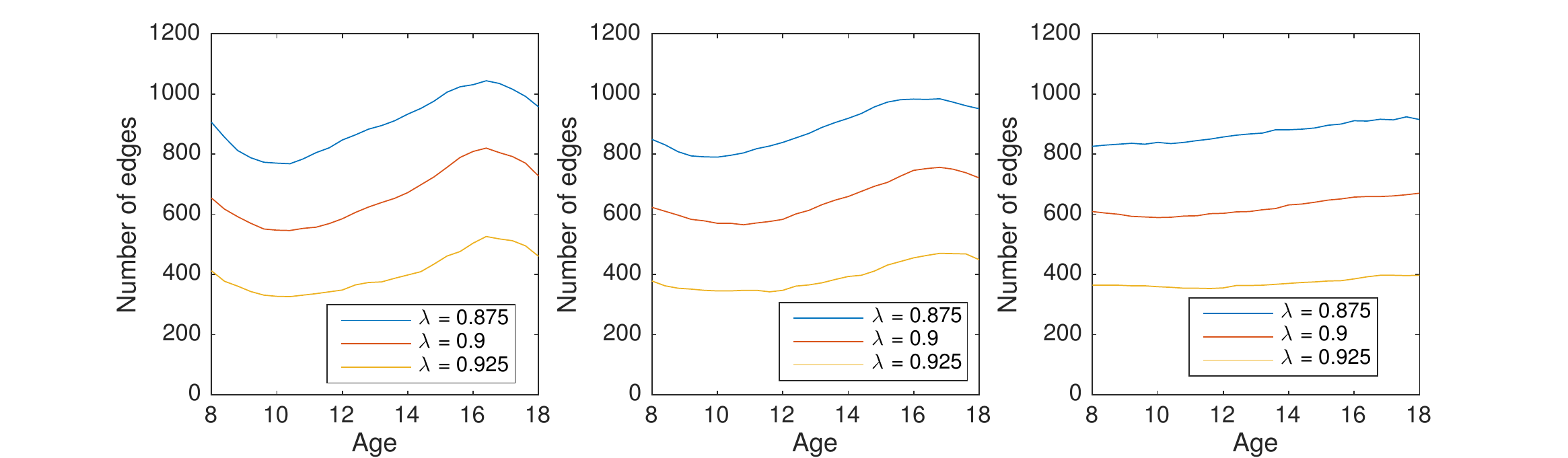}
\caption{Non-additive method of \cite{monti2014estimating} (optimally penalized version of \cite{qiu2016joint}).}
\end{subfigure}
\vfill
\begin{subfigure}[b]{\textwidth}
\centering
\includegraphics[width=5in]{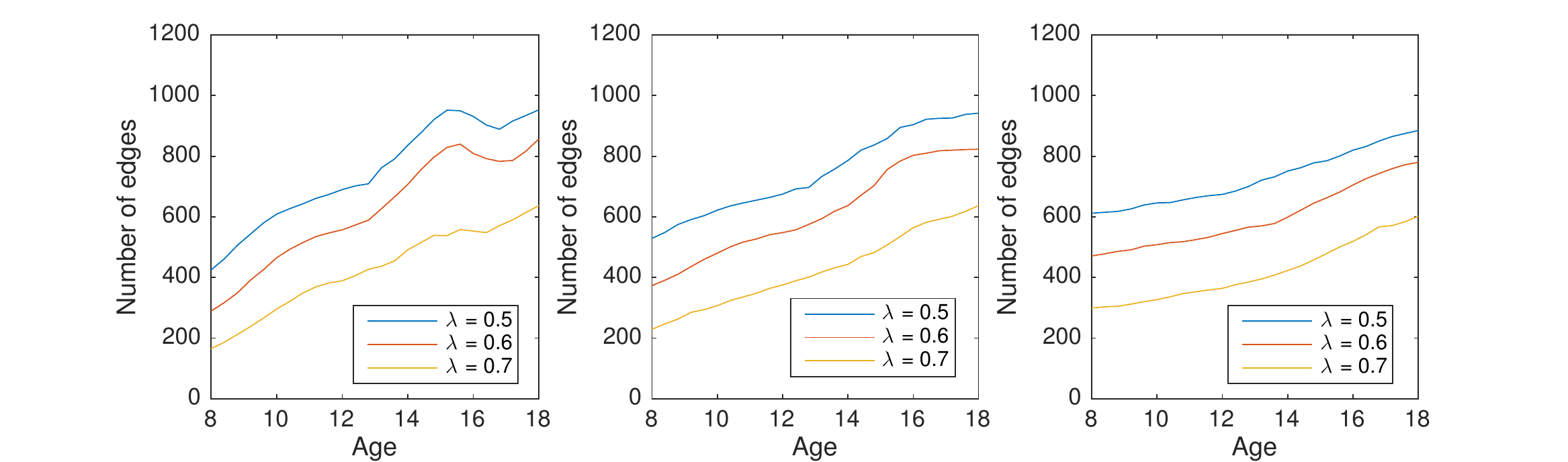}
\caption{Our proposed additive method, allowing for denoising of the time-correlated data.}
\end{subfigure}
\caption{Number of edges in the estimated $B^{-1}(t)$ graphical models across 90 brain regions as a function of age. Shown are results using three different values of the regularization parameter $\lambda$, and from left to right the kernel bandwidth parameter used is $h$ = 1.5, 2, and 3. Note the consistently increasing edge density in our estimate, corresponding to predictions of increased brain connectivity as the brain develops, leveling off in the late teenage years. Compare this to the method of \cite{monti2014estimating}, which successfully detects the trend in the years 11-14, but fails for other ages. 
}
\label{Fig:fEdges}
\end{figure}

\section{Conclusion}


In this work, we presented estimators for time-varying graphical models in the presence of time-correlated signals and noise. We revealed a bias-variance tradeoff scaling with the underlying rate of change, and proved strong single sample convergence results in high dimensions. We applied our methodology to an fMRI dataset, discovering meaningful temporal changes in functional connectivity, consistent with scientifically expected childhood growth and development.
\subsubsection*{Acknowledgement}
This work was supported in part by NSF under Grant DMS-1316731,
Elizabeth Caroline Crosby Research Award from the Advance Program at the
University of Michigan, and by AFOSR grant FA9550-13-1-0043.
\bibliographystyle{icml2017}
\bibliography{jmlr_bib}

\begin{thebibliography}{23}
\providecommand{\natexlab}[1]{#1}
\providecommand{\url}[1]{\texttt{#1}}
\expandafter\ifx\csname urlstyle\endcsname\relax
  \providecommand{\doi}[1]{doi: #1}\else
  \providecommand{\doi}{doi: \begingroup \urlstyle{rm}\Url}\fi

\bibitem[Arbabshirani et~al.(2014)Arbabshirani, Damaraju, Phlypo, Plis, Allen,
  Ma, Mathalon, Preda, Vaidya, and Adali]{arbabshirani2014impact}
Arbabshirani, M., Damaraju, E., Phlypo, R., Plis, S., Allen, E., Ma, S.,
  Mathalon, D., Preda, A., Vaidya, J., and Adali, T.
\newblock Impact of autocorrelation on functional connectivity.
\newblock \emph{Neuroimage}, 102:\penalty0 294--308, 2014.

\bibitem[Biswal et~al.(2010)Biswal, Mennes, Zuo, Gohel, Kelly, Smith, Beckmann,
  Adelstein, Buckner, and Colcombe]{biswal2010toward}
Biswal, B., Mennes, M., Zuo, X., Gohel, S., Kelly, C., Smith, S., Beckmann, C.,
  Adelstein, J., Buckner, R., and Colcombe, S.
\newblock Toward discovery science of human brain function.
\newblock \emph{Proceedings of the National Academy of Sciences}, 107\penalty0
  (10):\penalty0 4734--4739, 2010.

\bibitem[Boyd et~al.(2011)Boyd, Parikh, Chu, Peleato, and Eckstein]{Boyd:2011}
Boyd, S., Parikh, N., Chu, E., Peleato, B., and Eckstein, J.
\newblock Distributed optimization and statistical learning via {ADMM}.
\newblock \emph{Foundations and Trends{\textregistered} in Machine Learning},
  3\penalty0 (1):\penalty0 1--122, 2011.

\bibitem[Calhoun et~al.(2014)Calhoun, Miller, Pearlson, and
  Adal{\i}]{calhoun2014chronnectome}
Calhoun, V., Miller, R., Pearlson, G., and Adal{\i}, T.
\newblock The chronnectome: time-varying connectivity networks as the next
  frontier in f{MRI} data discovery.
\newblock \emph{Neuron}, 84\penalty0 (2):\penalty0 262--274, 2014.

\bibitem[Carvalho et~al.(2007)Carvalho, West, et~al.]{carvalho2007dynamic}
Carvalho, C., West, M., et~al.
\newblock Dynamic matrix-variate graphical models.
\newblock \emph{Bayesian analysis}, 2\penalty0 (1):\penalty0 69--97, 2007.

\bibitem[Chang \& Glover(2010)Chang and Glover]{chang2010time}
Chang, C. and Glover, G.
\newblock Time--frequency dynamics of resting-state brain connectivity measured
  with fmri.
\newblock \emph{Neuroimage}, 50\penalty0 (1):\penalty0 81--98, 2010.

\bibitem[Chen et~al.(2015)Chen, Liu, Yang, Xu, Lee, Lindquist, Caffo, and
  Vogelstein]{chen2015m}
Chen, S., Liu, K., Yang, Y., Xu, Y., Lee, S., Lindquist, M., Caffo, B., and
  Vogelstein, J.
\newblock An m-estimator for reduced-rank high-dimensional linear dynamical
  system identification.
\newblock \emph{arXiv:1509.03927}, 2015.

\bibitem[Cressie(2015)]{cressie2015statistics}
Cressie, N.
\newblock \emph{Statistics for spatial data}.
\newblock John Wiley \& Sons, 2015.

\bibitem[Greenewald \& Hero(2015)Greenewald and Hero]{greenewaldTSP}
Greenewald, K. and Hero, A.
\newblock Robust kronecker product {PCA} for spatio-temporal covariance
  estimation.
\newblock \emph{Signal Processing, IEEE Transactions on}, 63\penalty0
  (23):\penalty0 6368--6378, Dec 2015.

\bibitem[Huang et~al.(2010)Huang, Li, Sun, Ye, Fleisher, Wu, Chen, and
  Reiman]{huang2010learning}
Huang, S., Li, J., Sun, L., Ye, J., Fleisher, A., Wu, T., Chen, K., and Reiman,
  E.
\newblock Learning brain connectivity of alzheimer's disease by sparse inv.
  cov. est.
\newblock \emph{NeuroImage}, 50\penalty0 (3):\penalty0 935--949, 2010.

\bibitem[Kim et~al.(2015)Kim, Pan, Initiative, et~al.]{kim2015highly}
Kim, J., Pan, W., Initiative, Alzheimer's Disease~Neuroimaging, et~al.
\newblock Highly adaptive tests for group differences in brain functional
  connectivity.
\newblock \emph{NeuroImage: Clinical}, 9:\penalty0 625--639, 2015.

\bibitem[Liu \& Duyn(2013)Liu and Duyn]{liu2013time}
Liu, X. and Duyn, J.
\newblock Time-varying functional network information extracted from brief
  instances of spontaneous brain activity.
\newblock \emph{Proc. of the Natl. Academy of Sciences}, 110\penalty0
  (11):\penalty0 4392--4397, 2013.

\bibitem[Monti et~al.(2014)Monti, Hellyer, Sharp, Leech, Anagnostopoulos, and
  Montana]{monti2014estimating}
Monti, R., Hellyer, P., Sharp, D., Leech, R., Anagnostopoulos, C., and Montana,
  G.
\newblock Estimating time-varying brain conn. networks from f{MRI} time series.
\newblock \emph{NeuroImage}, 103:\penalty0 427--443, 2014.

\bibitem[Narayan et~al.(2015)Narayan, Allen, and Tomson]{narayan2015two}
Narayan, M., Allen, G., and Tomson, S.
\newblock Two sample inference for populations of graphical models with
  applications to functional connectivity.
\newblock \emph{arXiv preprint arXiv:1502.03853}, 2015.

\bibitem[Qiu et~al.(2016)Qiu, Han, Liu, and Caffo]{qiu2016joint}
Qiu, H., Han, F., Liu, H., and Caffo, B.
\newblock Joint estimation of multiple graphical models from high dimensional
  time series.
\newblock \emph{Journal of the Royal Statistical Society: Series B},
  78\penalty0 (2):\penalty0 487--504, 2016.

\bibitem[Rothman et~al.(2008)Rothman, Bickel, Levina, Zhu,
  et~al.]{rothman2008sparse}
Rothman, A., Bickel, P., Levina, E., Zhu, J., et~al.
\newblock Sparse permutation invariant covariance estimation.
\newblock \emph{Electronic Journal of Statistics}, 2:\penalty0 494--515, 2008.

\bibitem[Rudelson \& Zhou(2017)Rudelson and Zhou]{rudelson2015high}
Rudelson, M. and Zhou, S.
\newblock Errors-in-variables models with dependent measurements.
\newblock \emph{The Electronic Journal of Statistics}, 11\penalty0
  (1):\penalty0 1699--1797, 2017.

\bibitem[Tsiligkaridis \& Hero(2013)Tsiligkaridis and Hero]{tsiliArxiv}
Tsiligkaridis, T. and Hero, A.
\newblock Covariance estimation in high dimensions via kronecker product
  expansions.
\newblock \emph{IEEE Trans. on Sig. Proc.}, 61\penalty0 (21):\penalty0
  5347--5360, 2013.

\bibitem[Varoquaux et~al.(2010)Varoquaux, Gramfort, Poline, and
  Thirion]{varoquaux2010}
Varoquaux, G., Gramfort, A., Poline, J-B., and Thirion, B.
\newblock Brain covariance selection: better individual functional connectivity
  models using population prior.
\newblock \emph{Advances in Neural Information Processing Systems},
  23:\penalty0 2334--2342, 2010.

\bibitem[Wehbe et~al.(2014)Wehbe, Murphy, Talukdar, Fyshe, Ramdas, and
  Mitchell]{wehbe2014simultaneously}
Wehbe, L., Murphy, B., Talukdar, P., Fyshe, A., Ramdas, A., and Mitchell, T.
\newblock Simultaneously uncovering the patterns of brain regions involved in
  different story reading subprocesses.
\newblock \emph{PLOS ONE}, 9\penalty0 (11):\penalty0 e112575, 2014.

\bibitem[Zhou(2014)]{zhou2014gemini}
Zhou, S.
\newblock Gemini: Graph estimation with matrix variate normal instances.
\newblock \emph{The Annals of Statistics}, 42\penalty0 (2):\penalty0 532--562,
  2014.

\bibitem[Zhou et~al.(2010)Zhou, Lafferty, and Wasserman]{zhou:TV}
Zhou, S., Lafferty, J., and Wasserman, L.
\newblock Time varying undirected graphs.
\newblock \emph{Machine Learning}, 80\penalty0 (2-3):\penalty0 295--319, 2010.

\bibitem[Zhou et~al.(2011)Zhou, R{\"u}timann, Xu, and
  B{\"u}hlmann]{zhou2011high}
Zhou, S., R{\"u}timann, P., Xu, M., and B{\"u}hlmann, P.
\newblock High-dimensional covariance estimation based on gaussian graphical
  models.
\newblock \emph{The Journal of Machine Learning Research}, 12:\penalty0
  2975--3026, 2011.

\end{thebibliography}

\section{Estimating covariance $A$ for fMRI data}
We estimate the temporal covariance $A$ for each subject using the voxels contained in the regions of interest. The results for several example subjects are shown in Figure \ref{Fig:fA}, along with the corresponding sample covariance $\frac{1}{n}X^T X$ and its eigenvalues. We choose $\tau_B$ as the lower ``asymptote" of the eigenvalues of $\frac{1}{n}X^T X$, as in the high sample regime it is an upper bound on $\tau_B$. 

The sample covariance matrix $\frac{1}{n}X^T X$ is significantly diagonally dominant, supporting our subtraction of a $\tau_B = \tr(B)/n$ scaled identity matrix. Note the evident sparsity in the inverse. 
Many of the sparsity patterns indicate local AR-type behavior as assumed in \citep{qiu2016joint}, but this pattern is not stationary, and in fact tends to group in blocks. Hence, the assumptions in \citep{qiu2016joint} do not fully capture the richness of the data. 

\begin{figure}[h]
\centering
\includegraphics[width=6in]{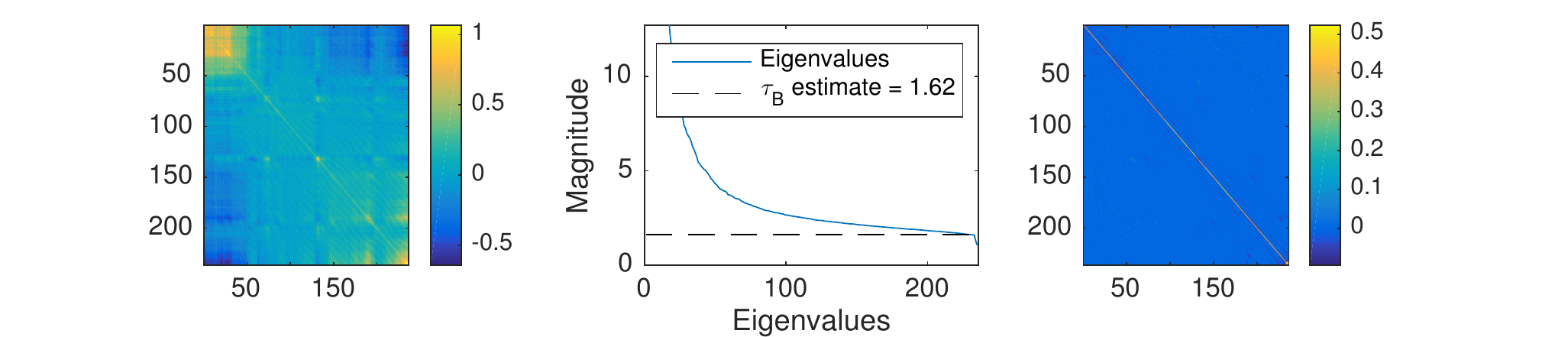}\\
\includegraphics[width=6in]{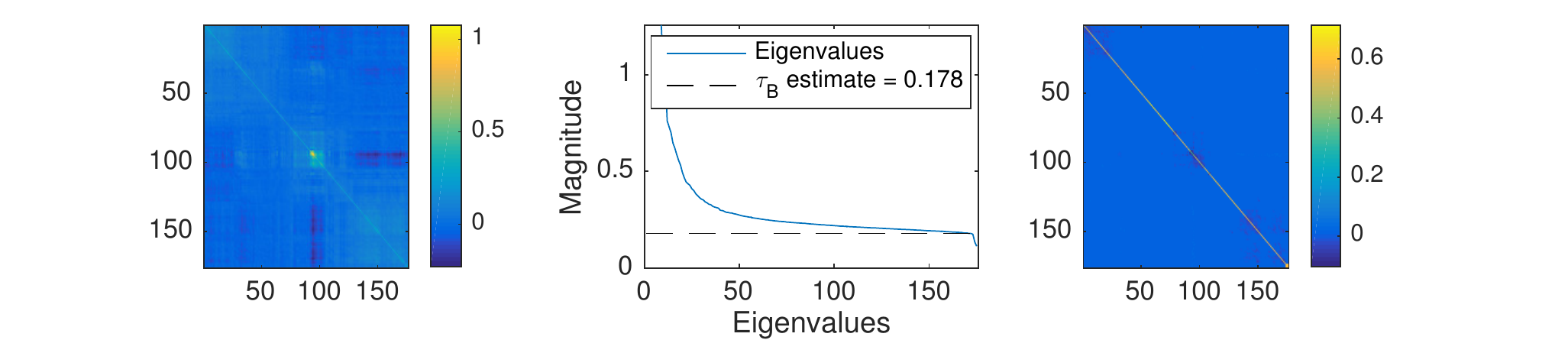}\\
\includegraphics[width=6in]{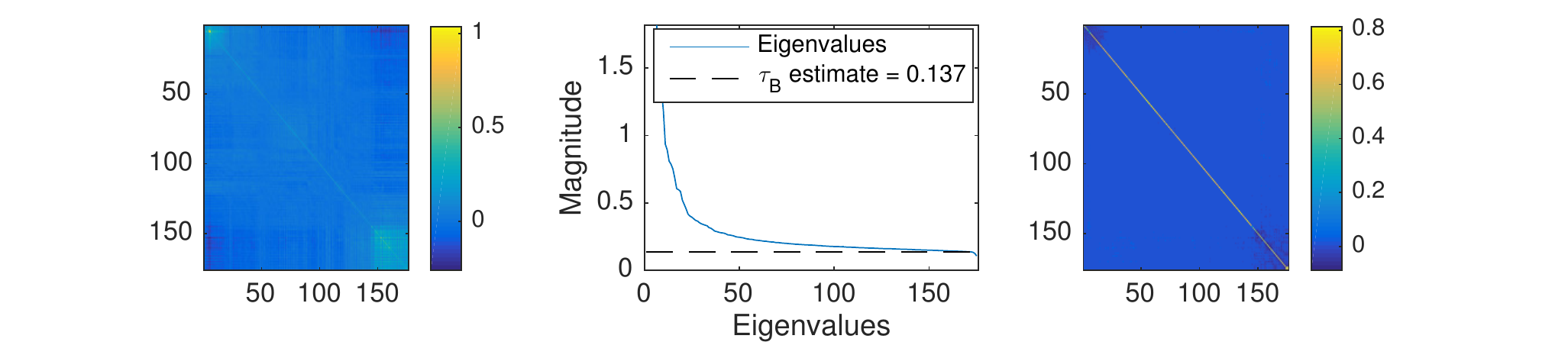}
\caption{Estimated $A$ covariance for three example subjects. Left: $A$ part sample covariance $\frac{1}{n}X^T X$. Center: Eigenvalues of $\frac{1}{n}X^T X$, showing estimate of $\tau_B$ factor. Right: Estimated $A^{-1}$ graphical model. Note the sparsity in the inverse, and that the eigenvalue spectra are consistent with the additive model.}
\label{Fig:fA}
\end{figure}

\subsection{Choosing $\mathrm{tr}(A)$}
\label{supp:tuning}

Due to the nonidentifiability of $\tr(A)$ and $\tr(B)$ when a single
copy of the data is observed, we assume that the $\tr(A)$ parameter is
either known~\citep{rudelson2015high}, or chosen as a tuning parameter.
If $A$ corresponds to a ``signal'' component in a physical system, then
as the mean signal strength $\tr(A)/m$ may be known by design or may
be estimated directly by running an experiment with the ``noise''
component involving the covariance parameter $B(t)$ to be
set at $0$ for a certain period of time.

The first step of relaxation is by assuming only one entry of $\mathrm{diag}(A)$,
for
example, $a_{11}$ is known. This is feasible when we assume we can observe
for a short period of time only $X_0 := Z_1 A^{1/2}$ so as to obtain the
knowledge of a
single element in $\mathrm{diag}(A)$. It is our conjecture that there is an
interesting tradeoff between the number of covariates in $X_0$ we are
allowed to observe with no measurement errors and the rate of
convergence we can obtain for estimating $A$ and $B$.
Moreover, even when no covariate $X_0$ is observed directly, we can rely
on
the recent progress
on high dimensional regression and signal reconstruction to help
establish theoretical limits on recovering $\tr(A)$ and $\tr(B(t))$, when
replicates
are available. For example, if a second sample of either $Z_B$ or $X_0 =
Z_1 A^{1/2}$ (cf. Equation (5)) is available, the corresponding
$\tr(B(t))$ or $\tr(A)$ can be estimated directly without needing to
specify any of $\tr(A)$ or $\tr(B(t))$.   We leave these as future
directions.


We now use two experimental settings to illustrate how
robust the estimation procedure for covariance $A$ is with respect to
the misspecification of $\mathrm{tr}(A)$.
We consider the case in which $m=400$, $n \in \{200,400\}$,
$\tau_A= \mathrm{tr}(A)/m=1$, $\tau_B=0.5$,
and $A$ is generated using the AR(1) or Star-Block, while $B$ follows
random ER graph.
We estimate the inverse of $A$ using $\hat{\tau}_A  \in \{0.4, 0.5,
\cdots,
1.4\}$ and $\lambda_A \in (0, 0.7)$, and use MCC to measure the
performance
of edge selection.


As shown in Figures \ref{Fig:Robust:A1}-\ref{Fig:Robust:A2},
we observe that when the topology is sophisticated (e.g., Star-Block) and
the misspecification error of $\hat{\tau}_A$ does not approach 0,
although joint tuning of $\lambda$ and $\hat{\tau}_A$
can not resolve the edge selection
inconsistency, as expected, although
it appears to give certain improvements.

When the topology is relatively simple, e.g., like the chain graph
correspondong to the  AR(1) model, the edge selection performance is
robust to the misspecification of $\hat{\tau}_A$; in this case, joint
tuning of $\lambda$ and $\tr(A)$ is not even necessary, as illustrated
in Figures \ref{Fig:Robust:A1}-\ref{Fig:Robust:A2}).

%

\begin{figure}[h]
\centering
\includegraphics[width=4.4in]{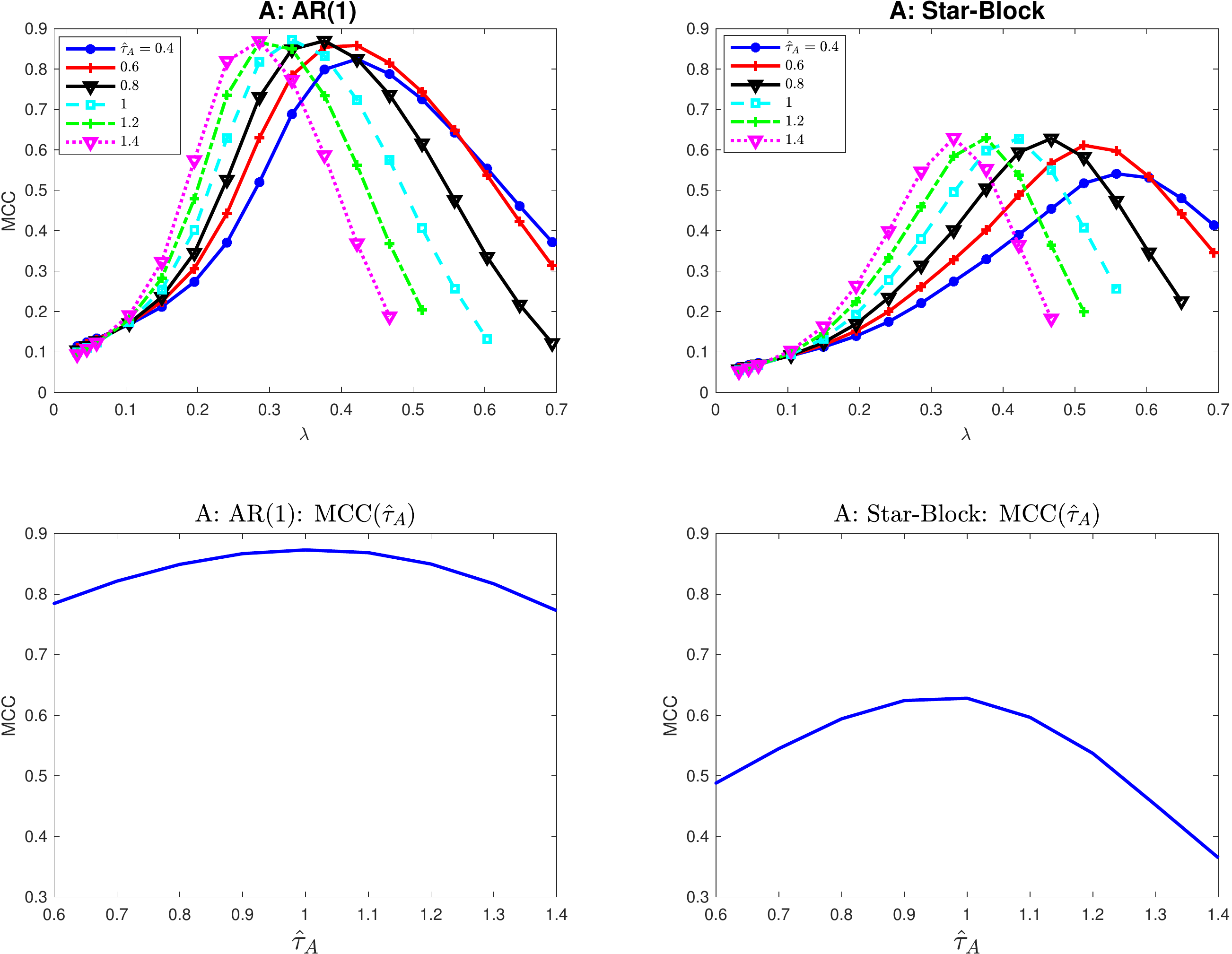}\\
\caption{The figures are MCC of the estimators $\hat{A}$ with respect to the true $A$ as a function of $\lambda$ for each $\hat{\mathrm{tr}}(A) \in \{0.6,  \cdots, 1.4\}$;
$m=400$ and $n=200$;
From left to right: $A$ is AR(1) and Star-Block model.
 }
\label{Fig:Robust:A1}
\end{figure}

\begin{figure}[h]
\centering
\includegraphics[width=4.4in]{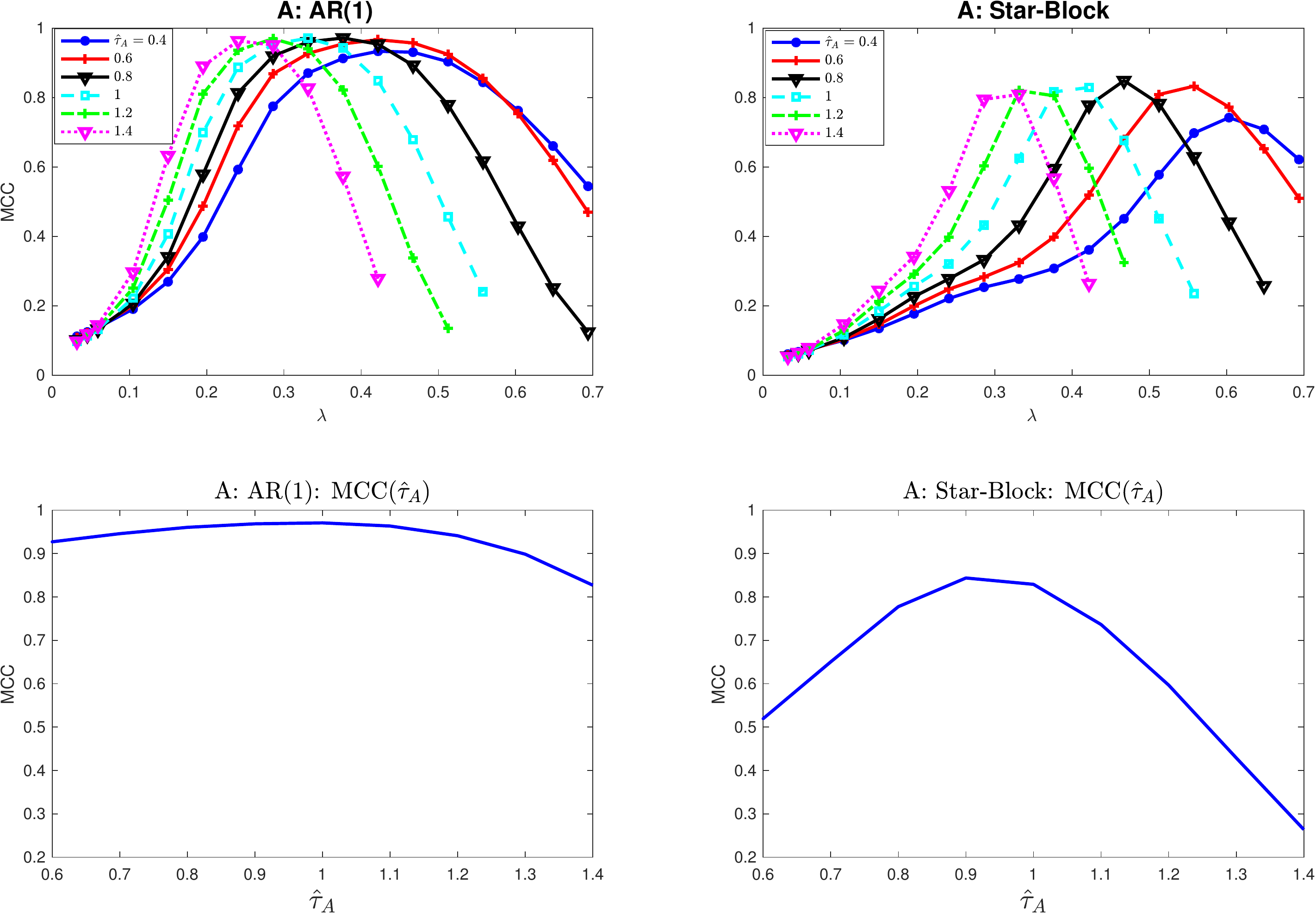}\\
\caption{The figures are MCC of the estimators $\hat{A}$ with respect to the true $A$ as a function of $\lambda$ for each $\hat{\mathrm{tr}}(A) \in \{0.6,  \cdots, 1.4\}$;
$m=400$ and $n=400$;
From left to right: $A$ is AR(1) and Star-Block model.
 }
\label{Fig:Robust:A2}
\end{figure}

\section{Additional analysis of ADHD-200 data}
In addition to the plots showing the relationship between brain connectivity and age in the main text, in this section we reproduce those plots using only healthy subjects and only using ADHD subjects. The results are shown in Figure \ref{Fig:fEdgesDiag}.

Observe that the leftmost plots (those with the narrowest kernel) are quite rough. This is caused by the reduction in the number of available subjects as compared to the plots in the main text that used all the subjects, increasing the effect of the noise. Shown in Figure \ref{Fig:fHist} are the age histograms for all patients, healthy patients only, and ADHD patients. Note the nonsmooth age ranges in the plots of Figure \ref{Fig:fEdgesDiag} correspond to regions with fewer available subjects, as expected. 

While the sample size is such that we cannot make definitive conclusions on the exact form of the differences between healthy and ADHD brain development, the fact that we observe significant differences is not surprising given the nature of ADHD and its effects on childhood development. In particular, note the lower rate of development among teenage ADHD subjects as opposed to healthy teenage subjects. We hypothesize that this corresponds to the common observation that ADHD patients face additional developmental challenges in the teenage years. 

\begin{figure}[h]
\centering
\begin{subfigure}[b]{\textwidth}
\centering
\includegraphics[width=5in]{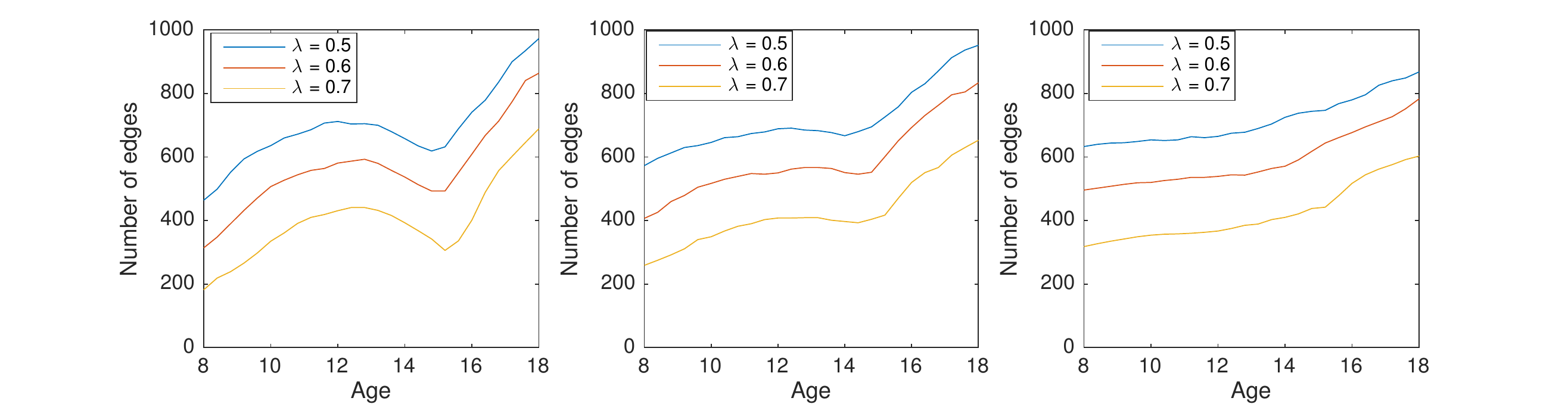}
\caption{Our proposed additive method, trained on healthy subjects only.}
\end{subfigure}
\vfill
\begin{subfigure}[b]{\textwidth}
\centering
\includegraphics[width=5in]{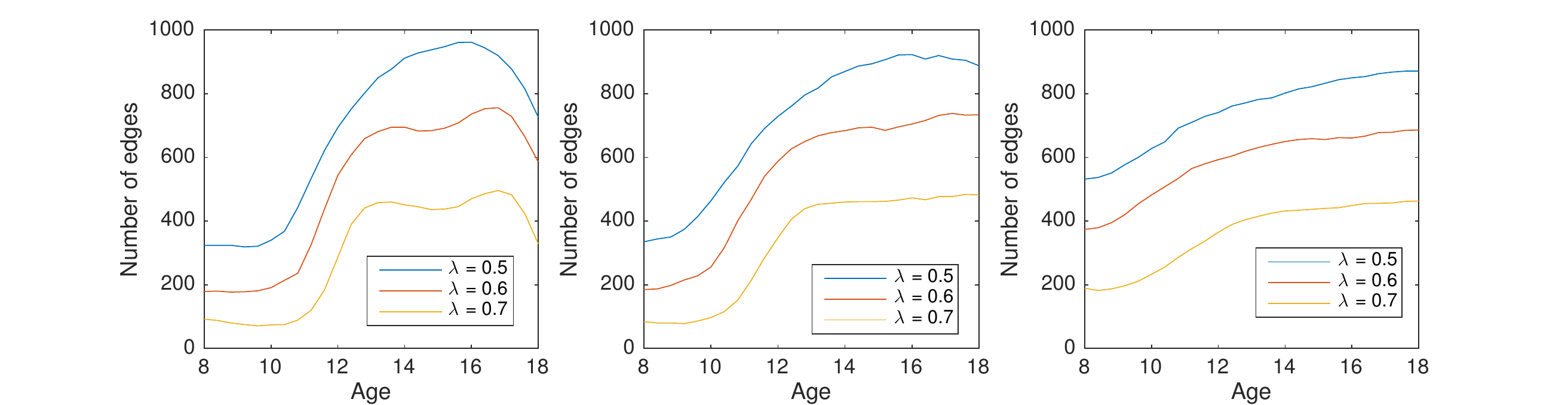}
\caption{Our proposed additive method, trained on ADHD subjects only.}
\end{subfigure}
\caption{Number of edges in the estimated $B^{-1}(t)$ graphical models across 90 brain regions as a function of age. Shown are results using three different values of the regularization parameter $\lambda$, and from left to right the kernel bandwidth parameter used is $h$ = 1.5, 2, and 3 for both methods. Note the consistently increasing edge density in our estimate, corresponding to predictions of increased brain connectivity as the brain develops, and the difference in teenage development rates between healthy and ADHD patients. 
}
\label{Fig:fEdgesDiag}
\end{figure}
\begin{figure}[h]
\centering
\begin{subfigure}[b]{.3\textwidth}
\centering
\includegraphics[width=1.75in]{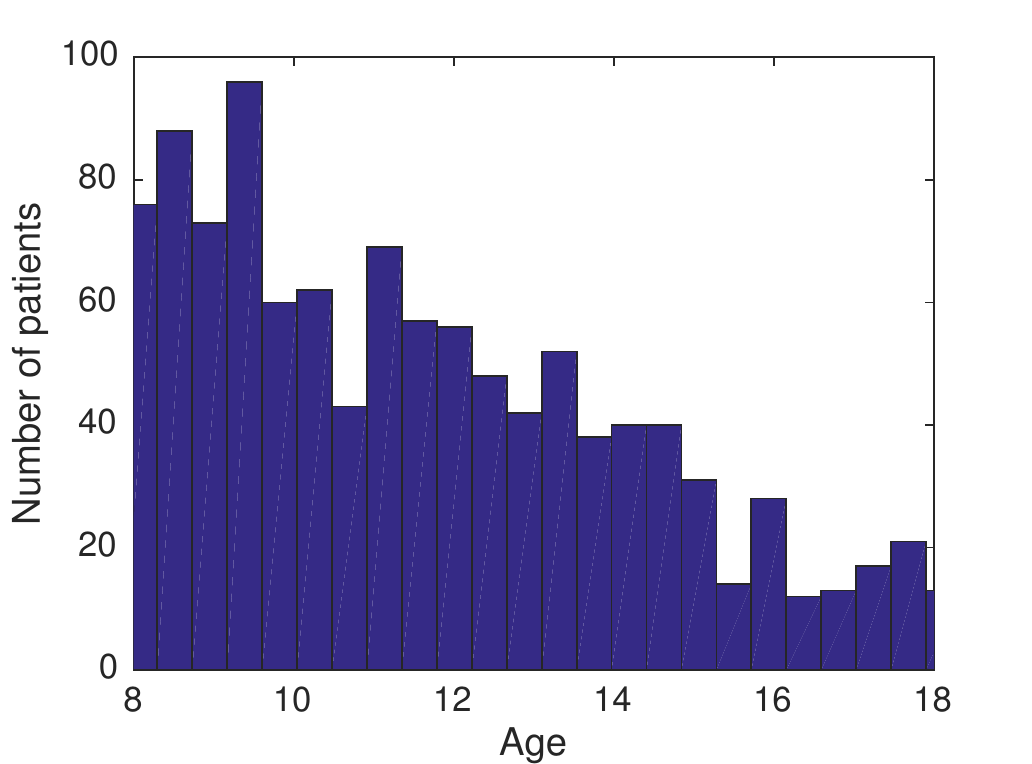}
\caption{Histogram of the ages of all subjects in the ADHD-200 dataset.}
\end{subfigure}
\hfill
\begin{subfigure}[b]{.3\textwidth}
\centering
\includegraphics[width=1.75in]{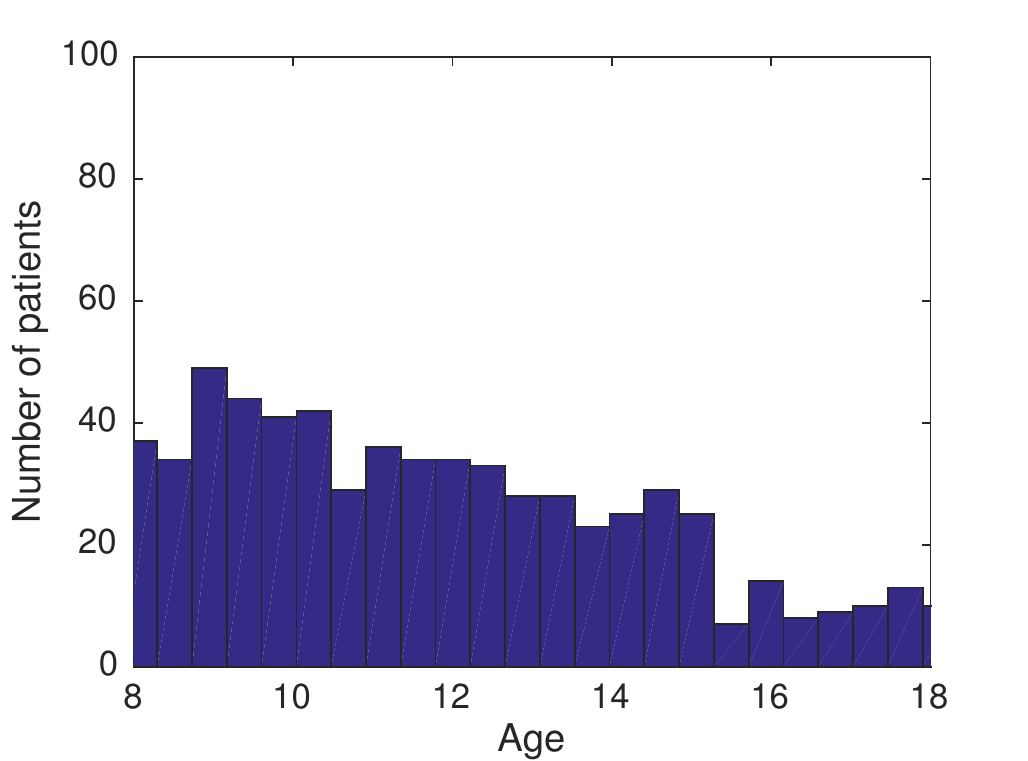}
\caption{Histogram of the ages of healthy subjects in the ADHD-200 dataset.}
\end{subfigure}
\hfill
\begin{subfigure}[b]{.3\textwidth}
\centering
\includegraphics[width=1.75in]{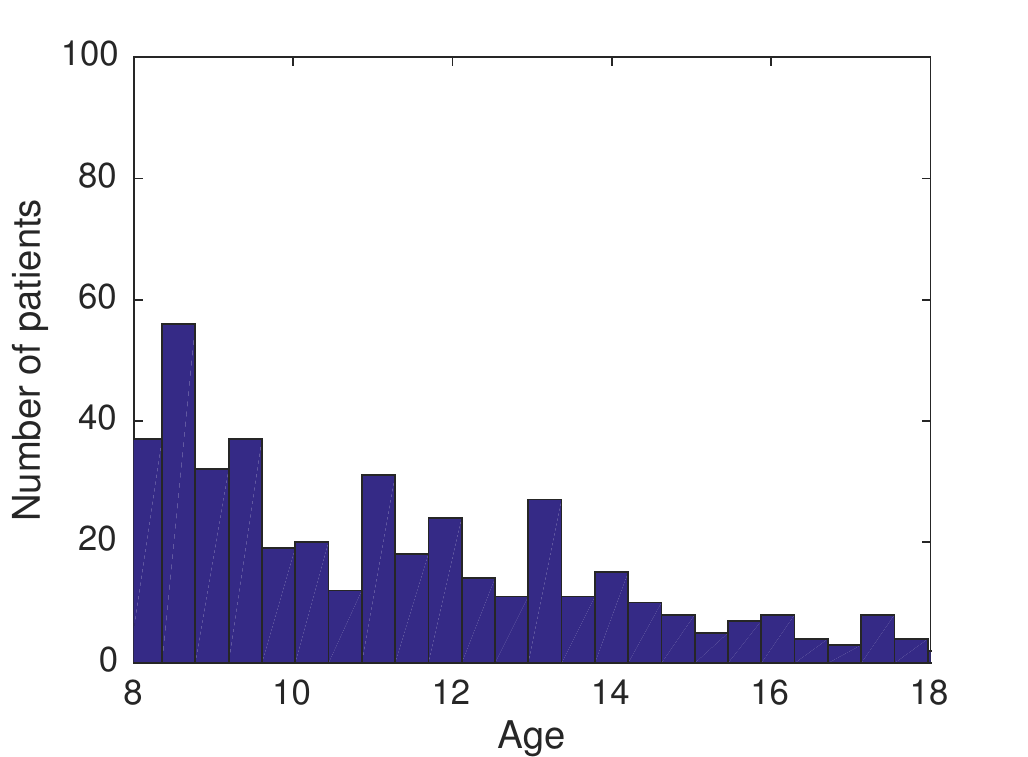}
\caption{Histogram of the ages of subjects diagnosed with ADHD in the ADHD-200 dataset.}
\end{subfigure}
\caption{Histograms showing the age distributions of all patients, healthy patients, and ADHD diagnosed patients in the ADHD-200 dataset. A set of subjects in the dataset have no diagnosis, these left out of both the healthy and ADHD groups. 
}
\label{Fig:fHist}
\end{figure}

\section{Technical assumptions}
\subsection{Assumptions}
In this section, we repeat the assumptions that we required in the main text.   
\begin{description}
      \item[Assumption A1] There exists a positive constant $c_A$ such that 
  \[
  \frac{1}{c_A} \leq \lambda_{\min}(A) \leq \lambda_{\max}(A) \leq c_A. 
  \]
\item[Assumption A2] The diagonal elements $A_{ii}$ are constant across all $i$, and the trace $\tr(B(t))$ is constant over time.
 \item[Assumption A3] $A^{-1}$ has at most $s_a  =o( n /\log m)$ nonzero off-diagonal elements.
  \item[Assumption B1] $B(t)$ is a symmetric and positive definite matrix for all $t$, and its entries have bounded second derivatives on $[0,1]$. 
   \item[Assumption B2]  There exists a positive constant $c_B$ such that for all $t  \in[0,1]$
                 \[
                 1/c_B \le  \lambda_{\min}(B(t)) \le \lambda_{\max}(B(t)) \le c_B.
                 \]

 \item[Assumption B3] $B(t)^{-1}$ has at most $s_b+n=o((m/\log m)^{2/3})$ nonzero off-diagonal elements.

\end{description}

Note that two conditions in Assumption A2 can be relaxed, we include them only to make statement of the estimators easier. In the nonnormalized cases the appropriate modifications are easy to derive.

We make the following assumptions on the smooth kernel function $K(\cdot)$ used to create the $\widehat{S}_m(t)$:
\begin{description}
  \item[Assumption K1] $K(\cdot)$ is non-negative, symmetric, twice differentiable, and has compact support $[-1,1]$.
  \item[Assumption K2] $\int K(u) du = 1$.
  \item[Assumption K3] $\int u^2 K(u) du < \infty$.
  \item[Assumption K4] $\sup_{u \in [-1,1]} K(u)  \le K_1$.
  \item[Assumption K5] $\sup_{u \in [-1,1]} K''\left(u/h\right) = O(h^{-4})$.
\end{description}
These are satisfied for most common smooth kernel functions, including the Gaussian kernel \citep{zhou:TV}. 
\section{Comparison of our method and Kronecker PCA}\label{supp:compare}
In this section, we compare our time-varying Kronecker sum model \eqref{eq:kronCov} to the sum of Kronecker products (KronPCA) model of \cite{tsiliArxiv,greenewaldTSP}
\[
\Sigma = \sum_{i=1}^r A_i \otimes B_i. 
\]
Both our method and KronPCA are a sum of Kronecker products, KronPCA is a more general model, but our method exploits sparsity while KronPCA cannot. Consider data generated from a time-varying Kronecker sum model \eqref{eq:kronCov}, where $A^{-1}\in \mathbf{R}^{60 \times 60}$ is a random ER graph and $B(t)^{-1}\in \mathbf{R}^{20\times 20}$ is a time-varying random ER graph as in Section \ref{sec:sims} in the main text. The sizes of $A, B$ where chosen to be relatively small since the computational complexity of KronPCA is $O(\min (m^6, n^6))$ (compared to the $O(m^3 + n^3)$ complexity of our method). Figure \ref{Fig:Comp} shows Frobenius norm results for our L1-penalized method, KronPCA, and for comparison, the baseline sample covariance. Note that due to the sparsity of the true model, our method performs significantly better than KronPCA, especially when the number of available replicates is small. Since in a time-varying setting the number of replicates is small or even 1, this is a significant advantage. Additionally, our method provides interpretable graph estimates, while the factors $A_i, B_i$ of KronPCA are nonsparse and not interpretable \cite{tsiliArxiv,greenewaldTSP}.

\begin{figure}[h]
\centering
\includegraphics[width=3.25in]{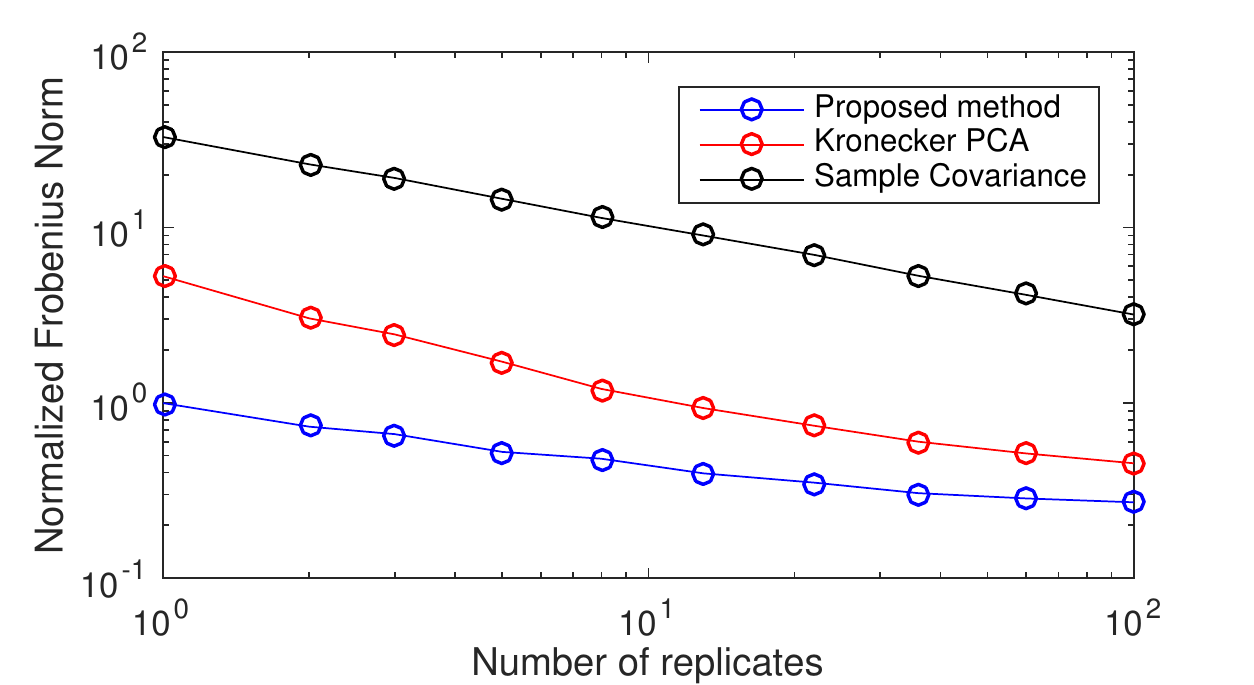}
\caption{Comparison of our method with KronPCA and the sample covariance. Shown is a logarithmic plot of the Frobenius norm error as a function of available replicates, with data generated using $A^{-1}\in \mathbf{R}^{60 \times 60}$, $B(t)^{-1}\in \mathbf{R}^{20\times 20}$ random ER graphs. }
\label{Fig:Comp}
\end{figure}

\section{Additional experiments using alternate graph topologies}
\label{supp:randgrid}

%
\subsection{$A$ star-block and MA}
In Figure \ref{Fig:BER2}, we repeat the experiments of main text Figure \ref{Fig:BER}, showing results for $A$ changed to a star-block graph (edge weights defined as for ER), and $A$ an moving average (MA) covariance (band width 15). The results confirm the trends found for the AR case. 
\begin{figure}[h]
\centering
\includegraphics[width=4.25in]{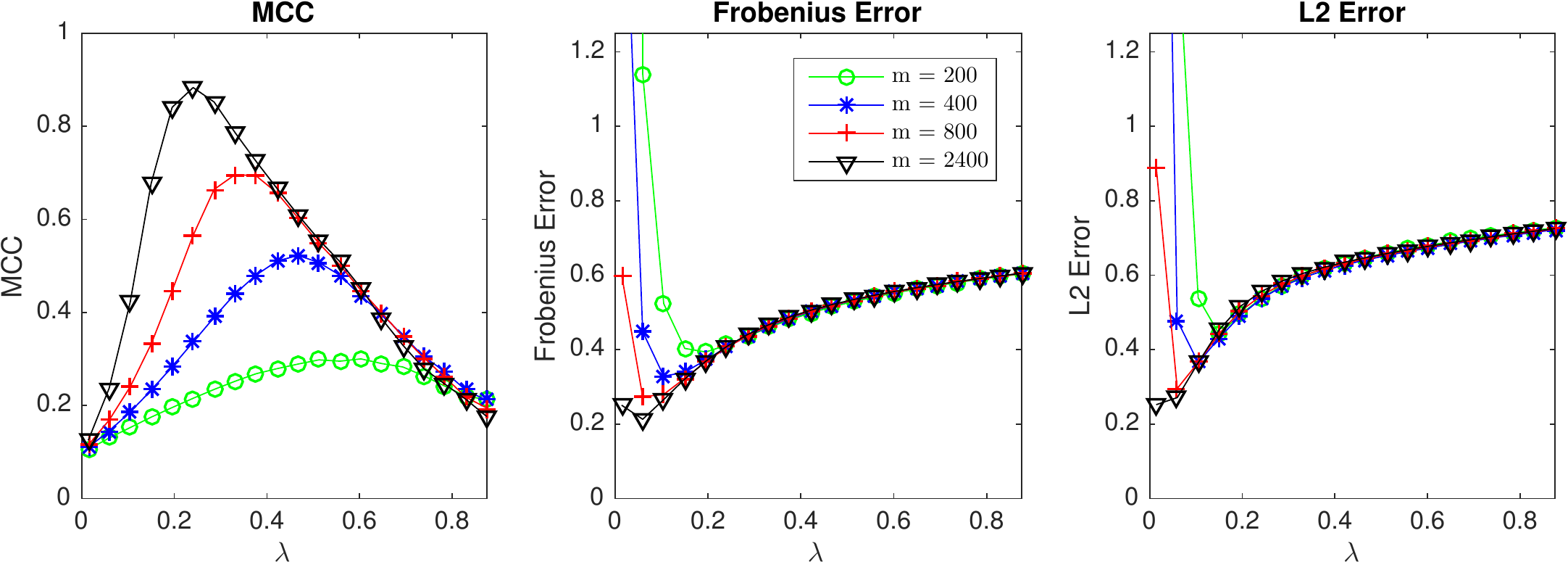}
\\\includegraphics[width=4.25in]{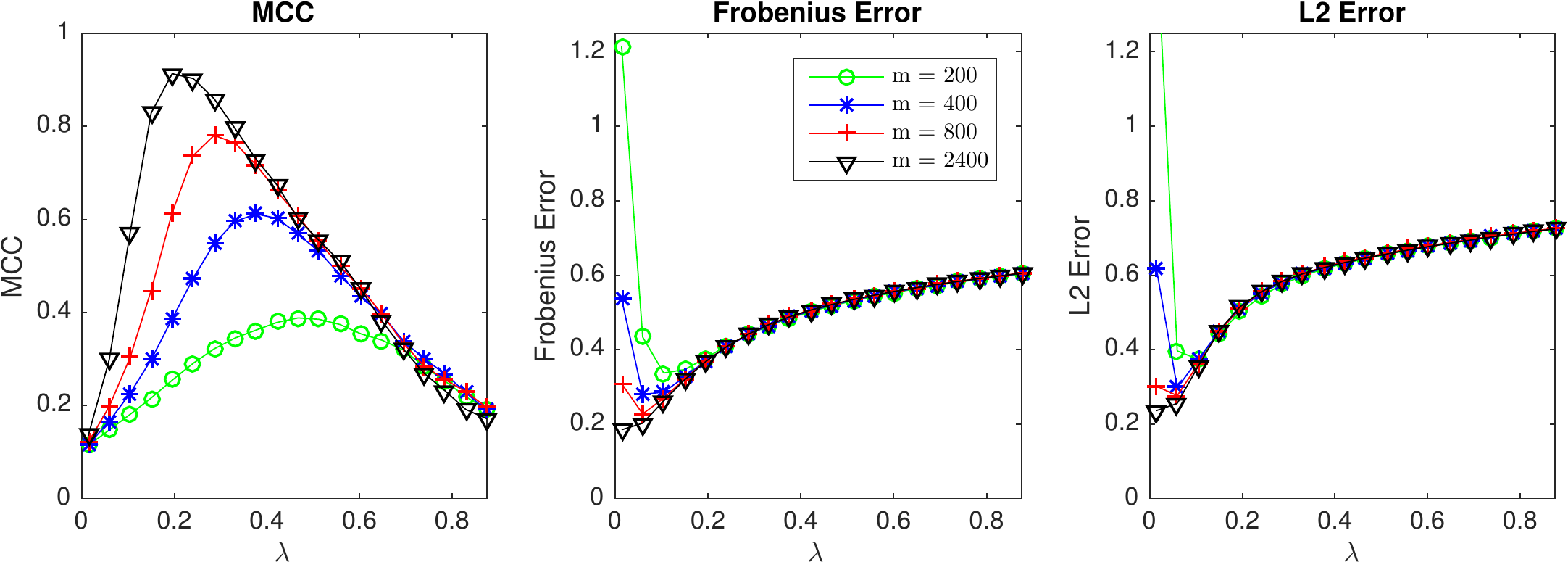}
\caption{MCC, Frobenius, and L2 norm error curves for $B$ a random ER graph and $n = 100$. From top to bottom: $A$ is star-block covariance, and MA covariance. }
\label{Fig:BER2}
\end{figure}
Similarly, in Figure \ref{Fig:St1}, we show the results for the $A$ estimator when $A$ is a star-block graph.
\begin{figure}[h]
\centering
\includegraphics[width=4.25in]{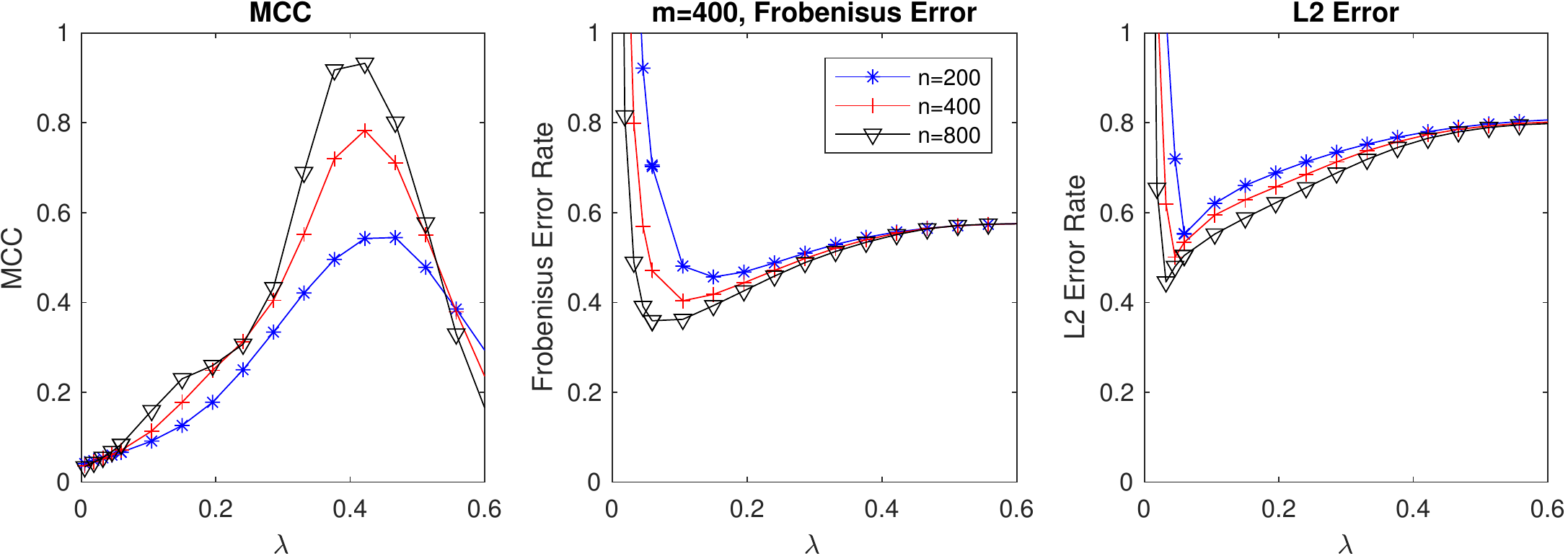}
\\\includegraphics[width=4.25in]{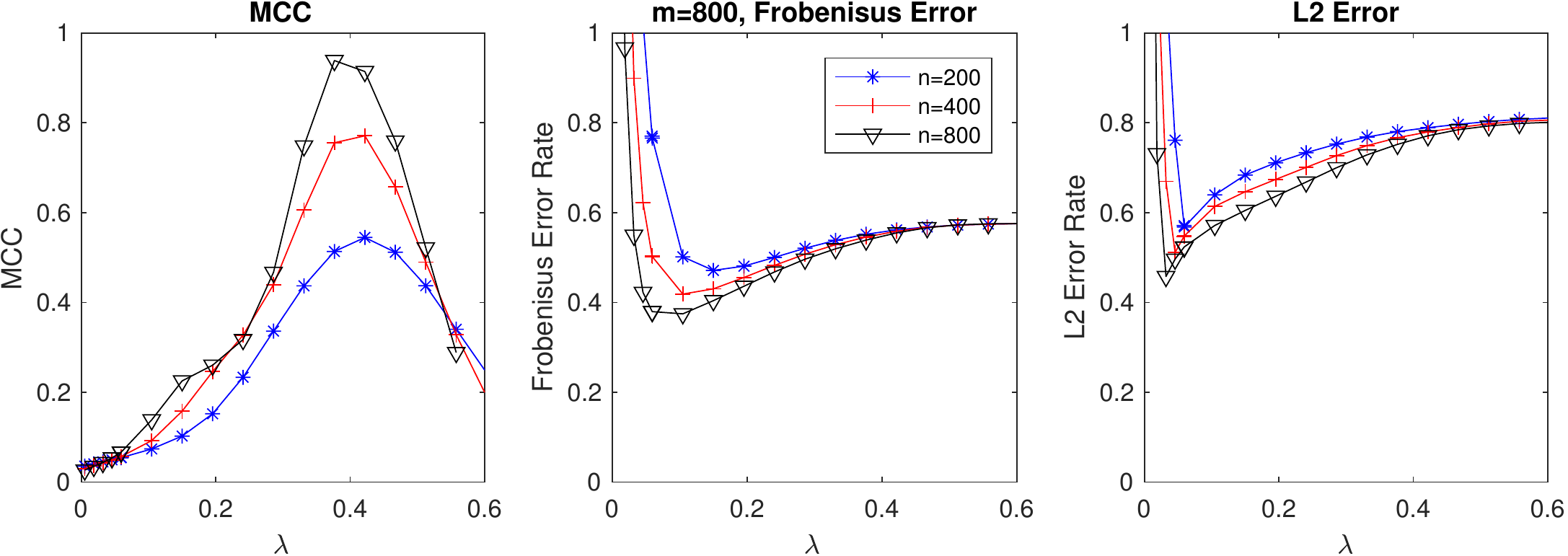}
\\\includegraphics[width=4.25in]{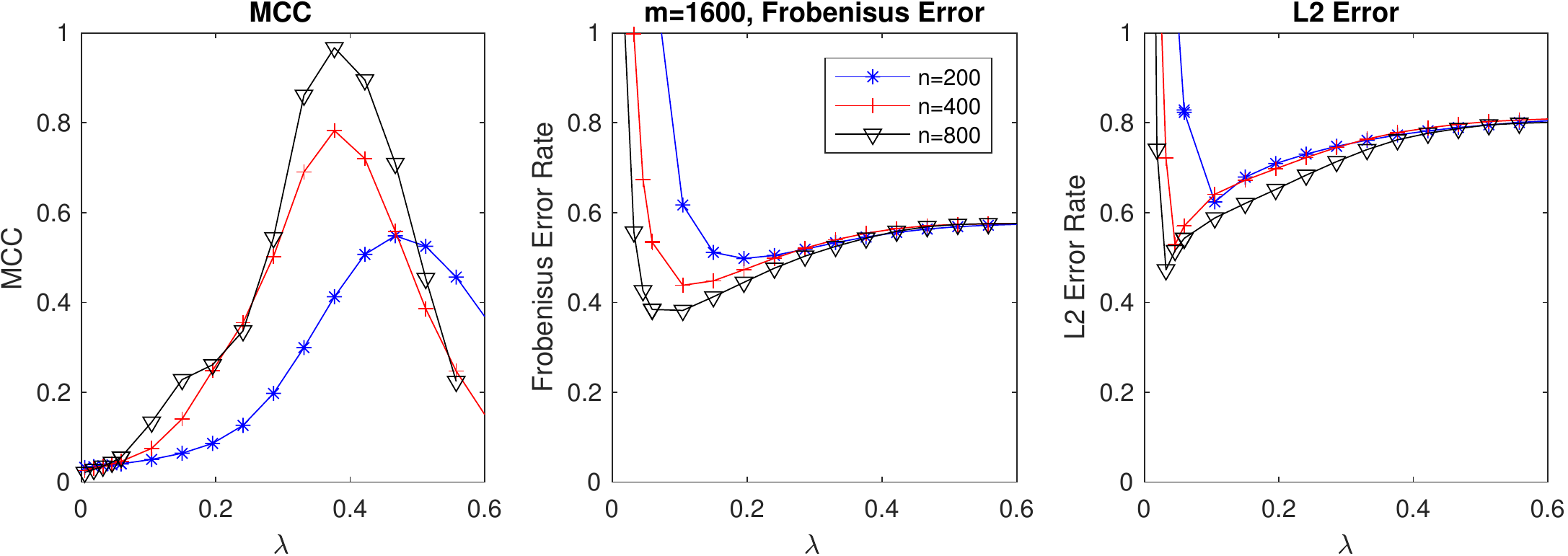}
\caption{MCC, Frobenius, and L2 norm error curves for $A$ a Star-Block graph when $B$ is a random ER graph. From top to bottom: $m=400$, $m=800$, and $m=1600$.}
\label{Fig:St1}
\end{figure}

\subsection{$B(t)$ random grid graph}

In this section, we use a random grid graph which is produced in the same way as the random ER graph, except edges are only allowed on between adjacent nodes in a square 2 dimensional grid (Figure \ref{Fig:GraphsGrid}). We replace the random ER model for $B(t)$ used in the main text with the random grid graph model, and repeat the main text experiments using this alternate $B(t)$ topology.

Random grid graph results for the experiments shown in main text Figure \ref{Fig:BER} are shown in Figure \ref{Fig:BGrid1}, showing similar results as expected. Similarly, random grid graph results for the experiments shown in Figure \ref{Fig:BER2} are shown in Figure \ref{Fig:BGrid2}. 

For the $A$ part, Figure \ref{Fig:ARg1} repeats the experiments of main text Figure \ref{Fig:AR1}, and Figure \ref{Fig:Stg1} repeats the experiments of Figure \ref{Fig:St1}, both using the random grid graph for $B(t)$.

\begin{figure}[h]
\centering
\includegraphics[width=4.4in]{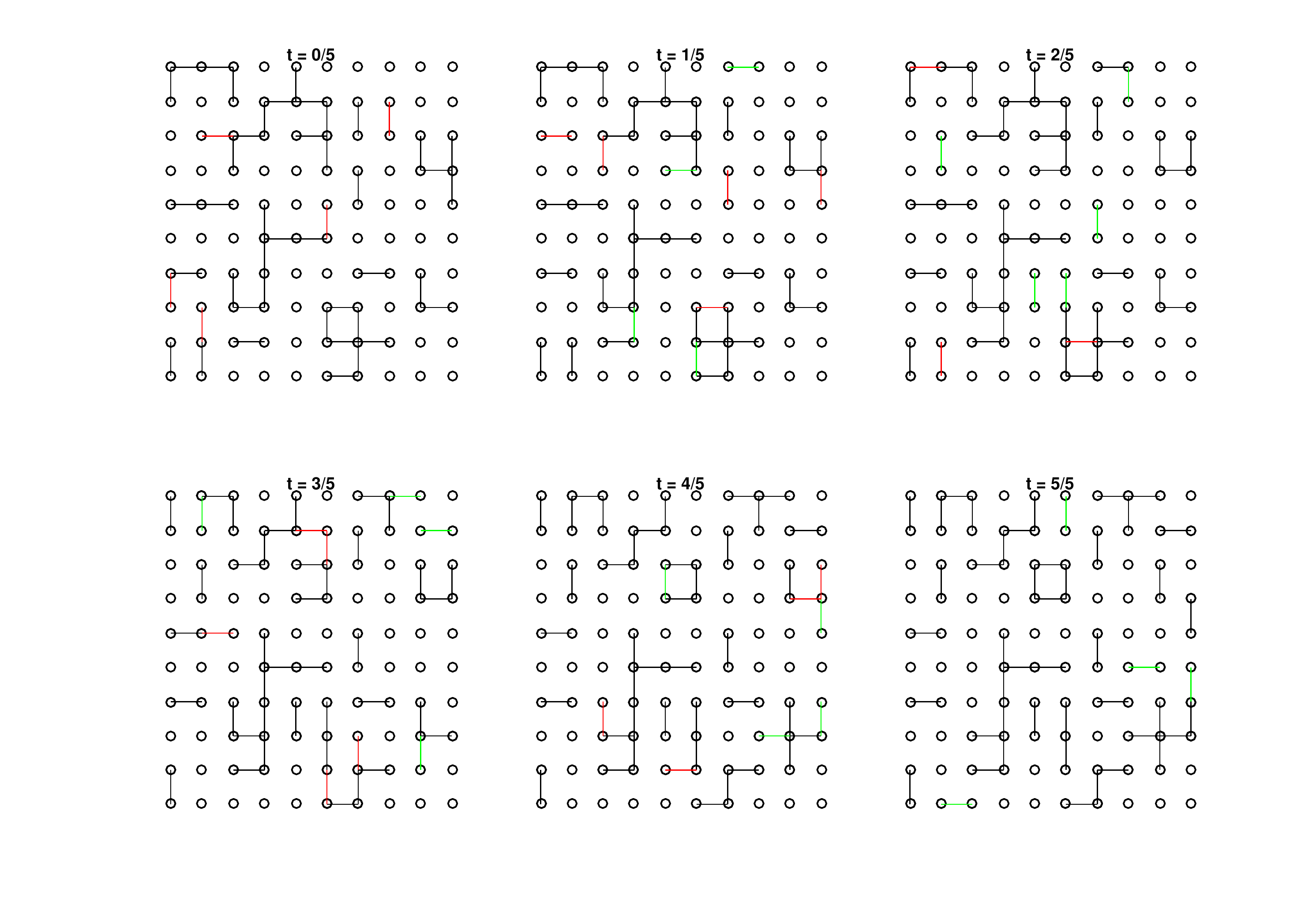}
\caption{Example sequence of $B^{-1}(t) = \Theta(t)$ random grid graphs used in the experiments. At each time point, the 50 edges connecting $n=100$ nodes are shown. Changes are indicated by red and green edges: red edges indicate edges that will be deleted in the next increment and green indicates new edges. }
\label{Fig:GraphsGrid}
\end{figure}
\begin{figure}[h]
\centering
\includegraphics[width=4.25in]{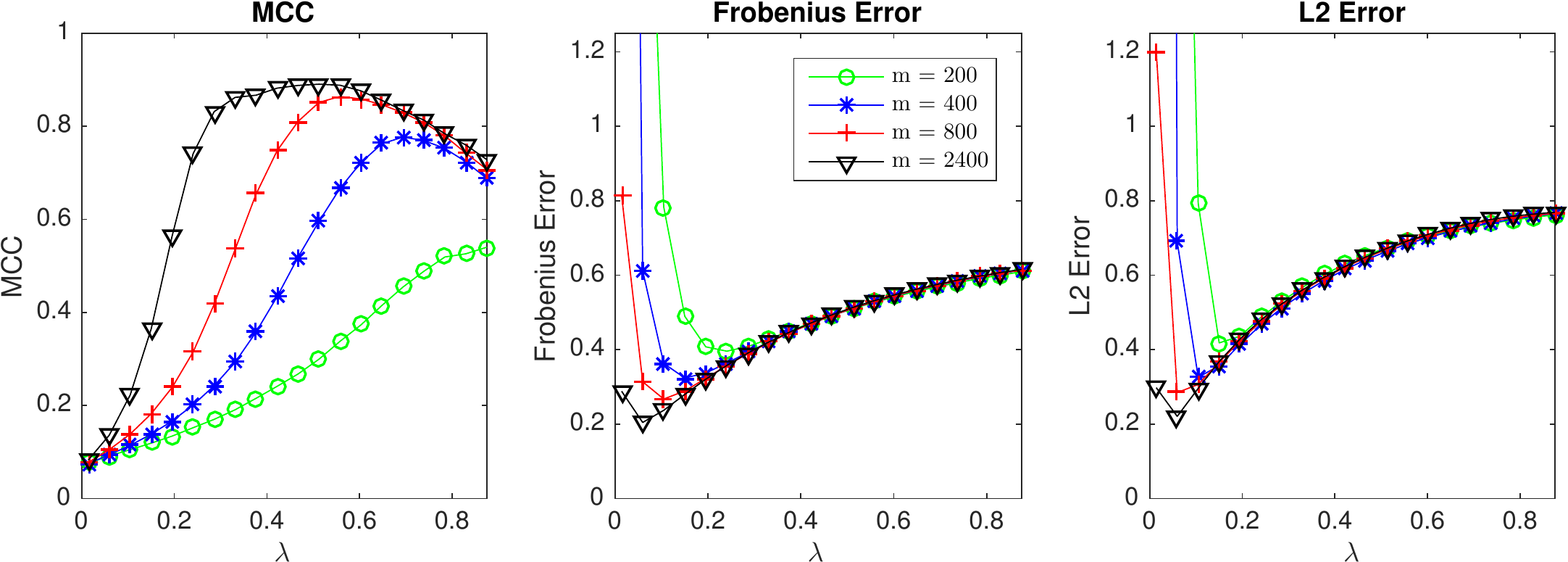}\\\includegraphics[width=4.25in]{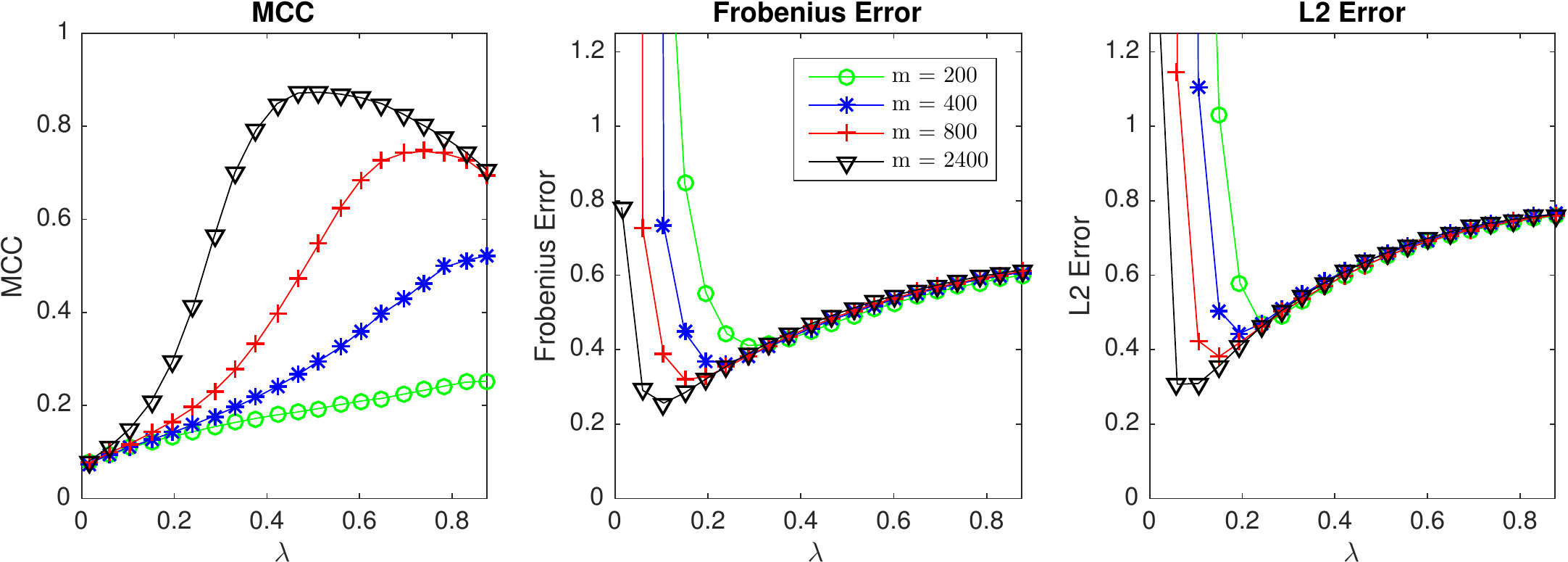}
\caption{MCC, Frobenius, and L2 norm error curves for $B$ a random grid graph and $n = 100$. From top to bottom: $A$ is AR covariance with $\rho = .5$, AR covariance with $\rho = .95$. }
\label{Fig:BGrid1}
\end{figure}
\begin{figure}[h]
\centering
\includegraphics[width=4.25in]{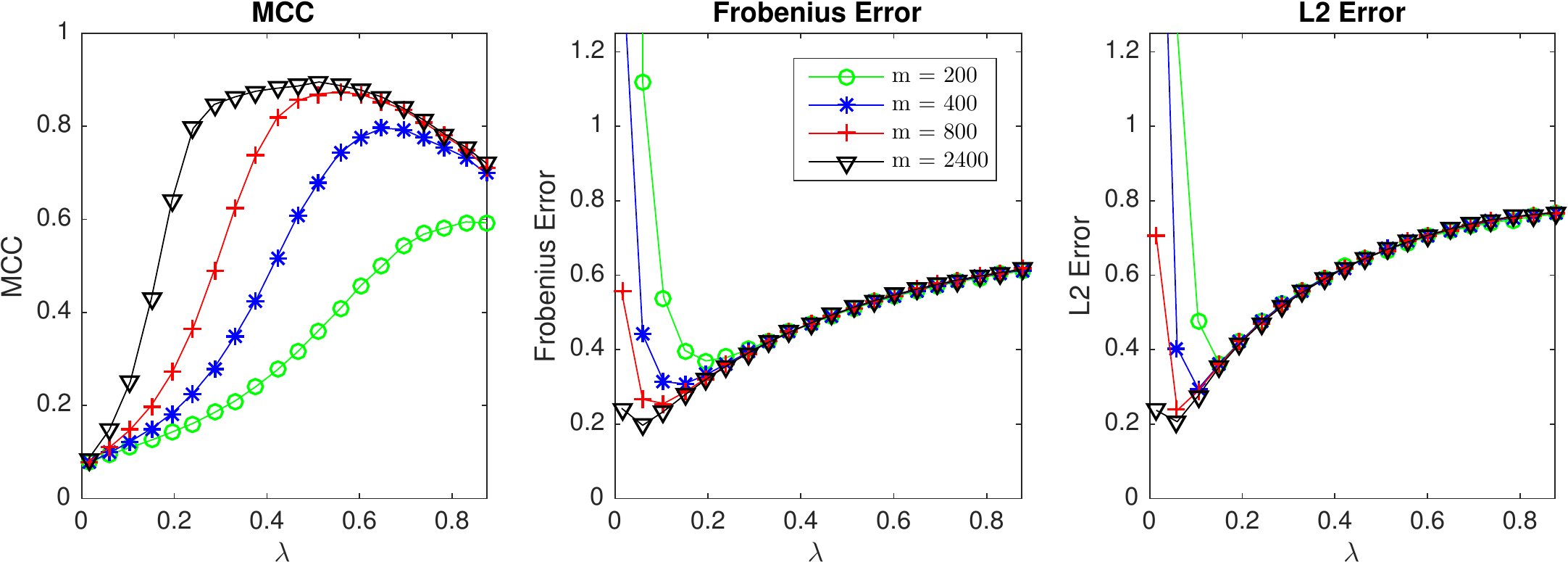}\\\includegraphics[width=4.25in]{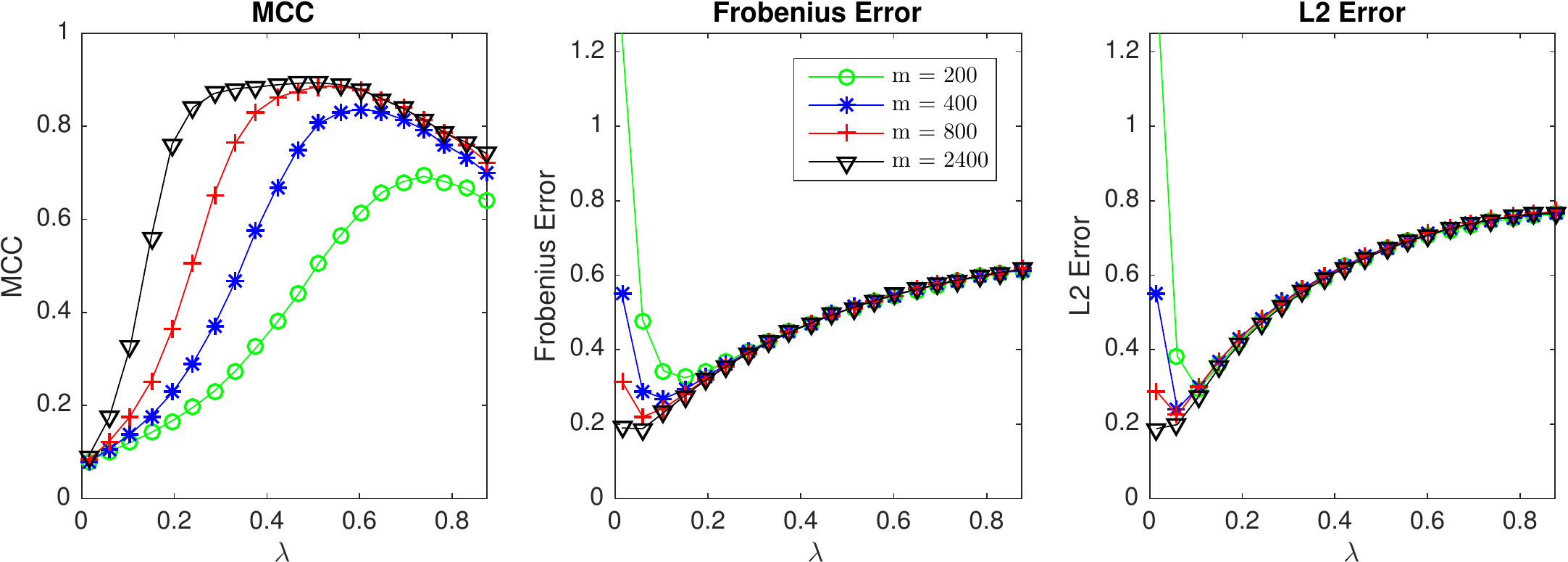}
\caption{MCC, Frobenius, and L2 norm error curves for $B$ a random grid graph and $n = 100$. From top to bottom: $A$ is star-block covariance, and MA covariance.}
\label{Fig:BGrid2}
\end{figure}
\begin{figure}[h]
\centering
\includegraphics[width=4.25in]{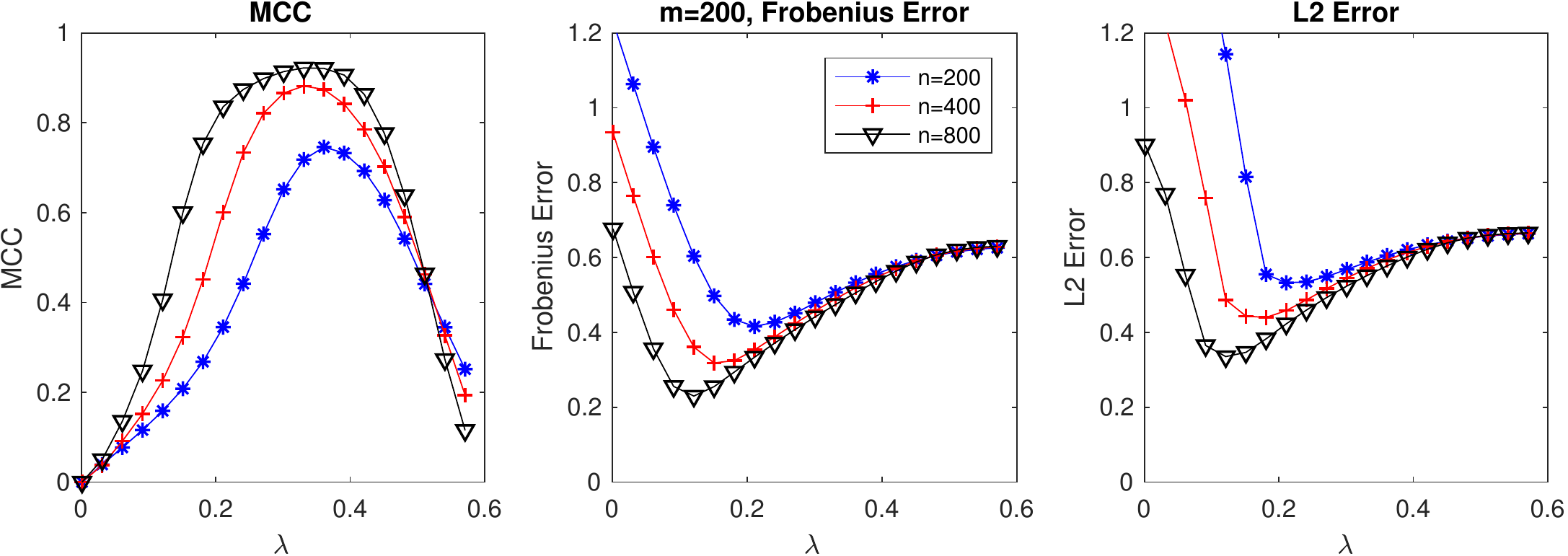}
\\\includegraphics[width=4.25in]{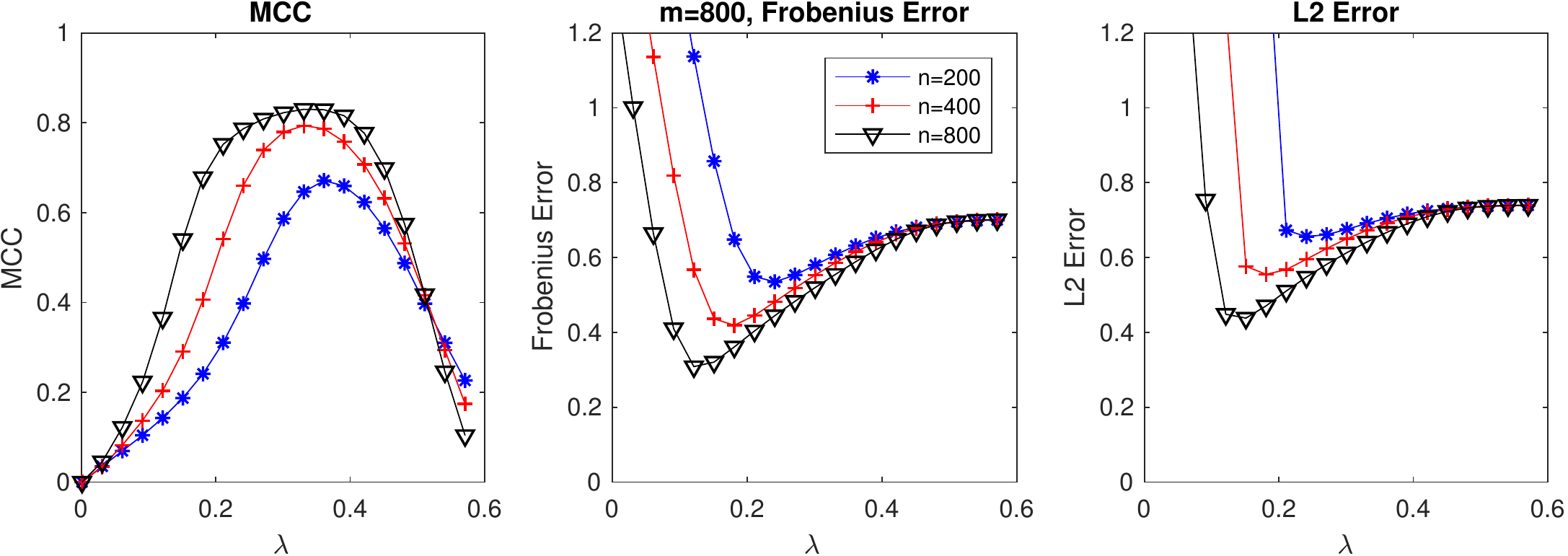}
\caption{MCC, Frobenius, and L2 norm error curves for $A$ a AR(1) with $\rho=0.5$ when $B$ is a random grid graph. From top to bottom: $m=200$ and $m=800$.}
\label{Fig:ARg1}
\end{figure}
\begin{figure}[h]
\centering
\includegraphics[width=4.25in]{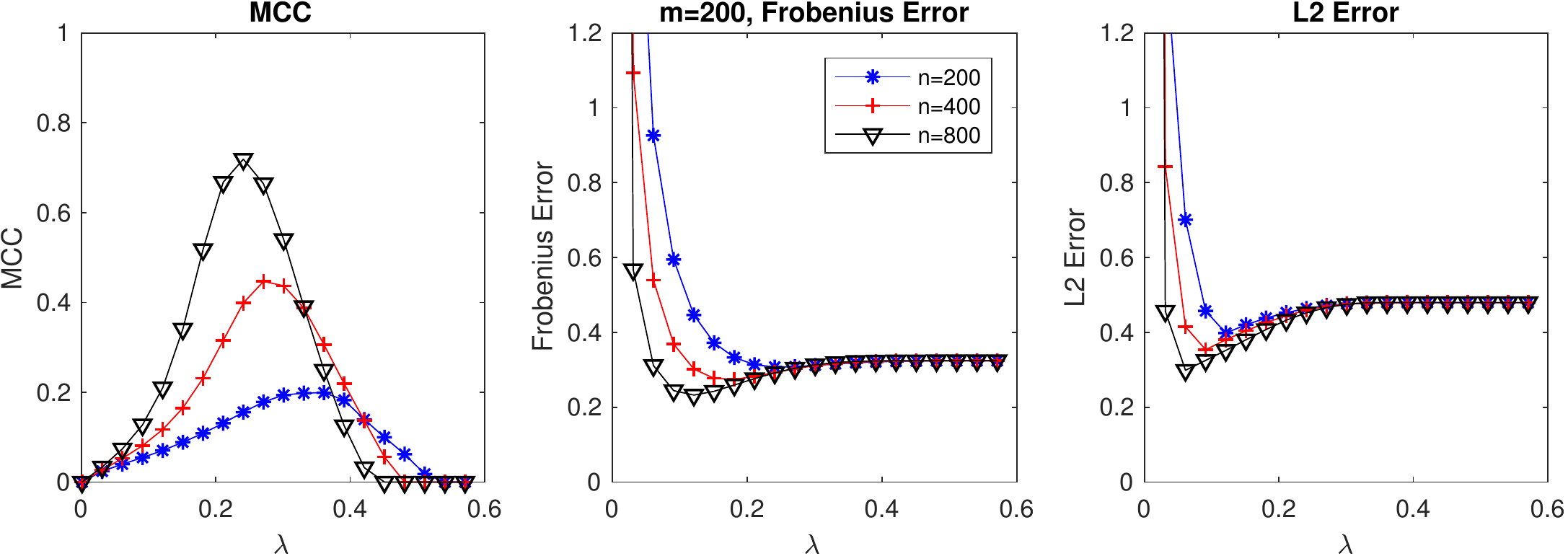}
\\\includegraphics[width=4.25in]{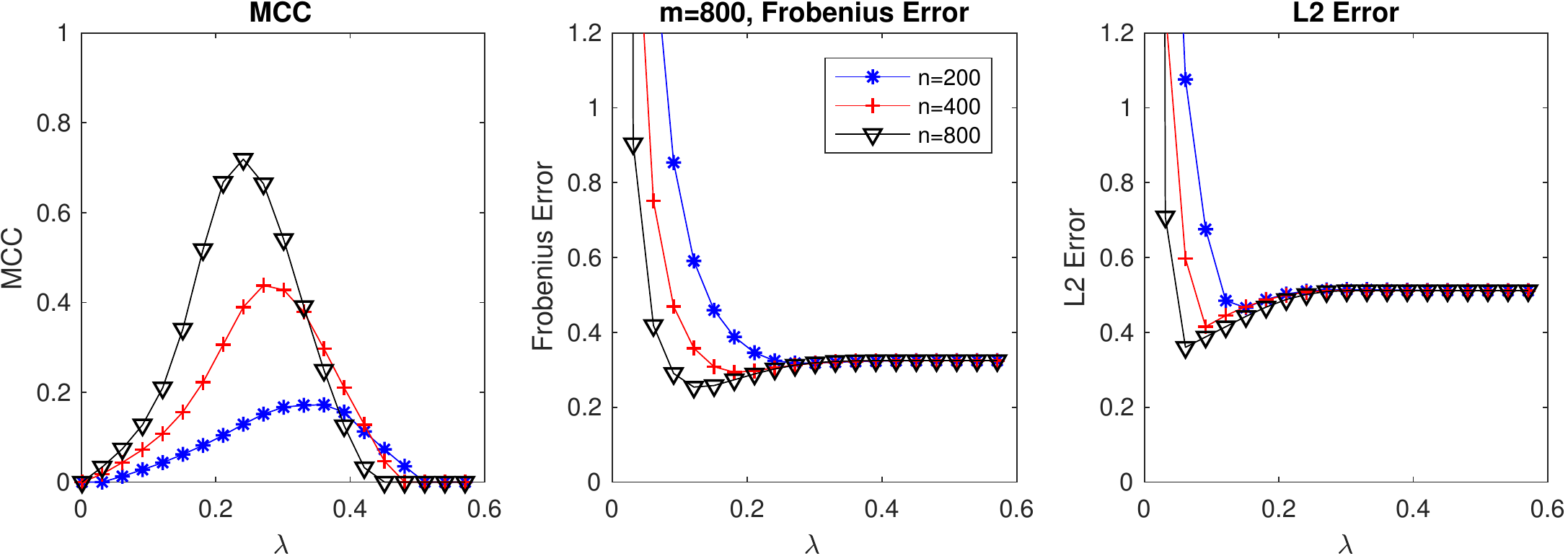}
\caption{MCC, Frobenius, and L2 norm error curves for $A$ a Star-Block graph when $B$ is a random grid graph. From top to bottom: $m=200$ and $m=800$.}
\label{Fig:Stg1}
\end{figure}

\section{Estimation error bound for ${B}$ part: Proof of Theorem \ref{thm}}
\subsection{Preliminary results}
\label{supp:prel}
Define $\tilde{B}(t_0)$ to be the expected value of the kernel-smoothed estimator $\hat{S}_m(t_0)$ at time $t_0$:
\begin{equation}
\tilde{B}(t_0) = E[\hat{S}_m (t_0)] = \sum_{i = 1}^m w_t(t_0) B(i/m).
\end{equation}

Using this notation, the estimator bias can be bounded via
\begin{lemma}[Bias]\label{lemma:bias1}
Suppose there exists $C>0$ such that $\max_{i,j} \sup_t | B''_{i,j}(t)| \leq C$. 
Then for a kernel $K(\cdot)$ satisfying assumptions (K1-K5)
we have
\begin{equation*}
\sup_{t\in[0,1]} \max_{i,j} | \tilde{B}(t_0,i,j) - B(t_0,i,j)| = O\left(h  + \frac{1}{m^2 h^5}\right).
\end{equation*}
\end{lemma}
This lemma is proved in Section \ref{app:kernel}.

The variance of the estimator $\hat{S}_m(t_0)$ can be bounded as:
\begin{lemma}[Variance]
\label{lemma:varPSD}
Suppose $mh > n$. 
Define event $\mathcal{A}$ 
\begin{align}
\max_{i,j} |\hat{S}_m(t_0, i,j) - \tilde{B}(t_0,i,j)| \leq C \sqrt{\frac{\log mh}{mh}},
\end{align}
for some $C > 0$.
Then 
$
\mathbb{P}(\mathcal{A}) \geq 1-\frac{c}{m^4h^4}. 
$
\end{lemma}
The proof of this result is based on an application of the Hanson-Wright inequality, and is found in the supplementary material.
We emphasize that this bound holds for both the diagonal and off-diagonal elements simultaneously. Using a similar approach, in Section \ref{app:posdef} we can also show that the estimator $\hat{S}_m(t_0)$ is positive definite with high probability:
\begin{lemma}[Positive definiteness]
\label{lemma:varPSD2}
Suppose $mh > n$. 
Define the event $\mathcal{B}$
\begin{align}
(1+\delta) \tilde{B}(t_0) \succ \hat{S}_m(t_0) \succ (1-\delta)\tilde{B}(t_0) \succ 0.
\end{align}
for some fixed $\sqrt{\frac{c\log mh}{mh}} \leq \delta < 1$.
Then 
$
\mathbb{P}(\mathcal{B}) \geq 1-\frac{c}{m^4h^4}. 
$
\end{lemma}

\subsection{Proof of Theorem \ref{thm}}
\label{Sec:lemmaVar}
In this section, we derive the elementwise bound on the estimator $\widehat{S}_m(t_0)$ of the spatial covariance $B(t_0)$ at time $t_0$, and show that it is positive definite with high probability. 
To obtain the elementwise bound, we will first bound the bias and variance of $\widehat{S}_m(t_0, i, j)$ and then combine the bounds.

\subsection{Estimator bias bound}
Recall that $\tilde{B}(t_0) = E[\hat{S}_m (t_0)]$.
By Lemma \ref{lemma:bias1} with $p = 1$ (proof in Section \ref{app:kernel}),
we have 
\[
\sup_{t_0} \max_{i,j} |\tilde{B}(t_0,i,j) - B({t_0, i, j})| = O\left(h + \frac{1}{m^2 h^5}\right).
\]


\subsection{Estimator variance bound (Lemma \ref{lemma:varPSD})}
\setcounter{theorem}{1}
\begin{lemma}[Variance]
\label{lemma:varPSD}
Suppose $mh > n$. 
Define event $\mathcal{A}$ 
\begin{align}
\max_{i,j} |\hat{S}_m(t_0, i,j) - \tilde{B}(t_0,i,j)| \leq C \sqrt{\frac{\log mh}{mh}},
\end{align}
for some $C > 0$.
Then 
\begin{align}
\mathbb{P}(\mathcal{A}) \geq 1-\frac{c}{m^4h^4}. 
\end{align}
\end{lemma}

\begin{proof}
Recall that
\begin{align}\label{eq:PlugInEst}
\widehat{S}_m(t_0) := \sum_{i = 1}^m w_{i}({t_0}) \left(X_i  X_i^T - \frac{\tr(A)}{m} I_n\right).
\end{align}
and
\begin{align*}
w_{i}(t_0) =  \frac{1}{mh} K\left(\frac{i/m-t_0}{h}\right). 
\end{align*}
Then
\begin{align*}
\widehat{S}_m(t_0) - \tilde{B}(t_0) &= \sum_{\ell= 1}^m w_{\ell}({t_0}) (X_\ell  X_\ell^T - E X_\ell  X_\ell^T).
\end{align*}
Let $\hat{S} = \mathrm{vec}(X) \mathrm{vec}(X)^T \in \mathbb{R}^{mn \times mn}$ be the overall sample covariance. Then observe that for fixed $i,j,t_0$ there exists a vector $w^{ij}(t_0)$ with $\sum_t w^{(ij)}_t (t_0)= 1$ and $\|w^{(ij)}(t_0)\|_2 \leq c/\sqrt{mh}$
such that
\[
\widehat{S}_m(t_0, i, j)- \tilde{B}(t_0,i,j) = [w^{(ij)}(t_0)]^T \mathrm{vec} (\hat{S}-\Sigma).
\]
By the triangle inequality, $\|\Sigma\| \leq \|A\| + \max_t \|B_t\|$, so $\|\Sigma\| \|w^{(ij)}(t_0)\|_2 \leq O(1/\sqrt{mh})$. 
We can thus apply Lemma \ref{lemma:var} in Section \ref{sec:conc}, giving for fixed $i,j$
\begin{align*}
\mathbb{P}( |\widehat{S}_m(t_0, i, j) - \tilde{B}(t_0,i,j)|  \geq \epsilon \|\Sigma\| \|w^{(ij}(t_0)\|_2 ) 
&=\mathbb{P}\left([w^{(ij)}(t_0)]^T \mathrm{vec} (\hat{S}-\Sigma)  \geq \epsilon C \sqrt{\frac{1}{mh}}\right) \\
&\leq 2\exp\left(-c \frac{\epsilon^2}{K^4}\right).
\end{align*}

Using the union bound over $i,j$ (cardinality $n^2$), the concentration bound Lemma \ref{lemma:var} implies
\begin{align*}
\mathbb{P}\left( \max_{i,j} |\widehat{S}_m(t_0, i, j) - \tilde{B}(t_0,i,j)|  \geq \epsilon C \sqrt{\frac{1}{mh}} \right) 
&\leq n^2 2\exp\left(-c \frac{\epsilon^2}{K^4}\right)\\
&= 2\exp\left(2 \log n -c \frac{\epsilon^2}{K^4}\right).
\end{align*}
Setting $\epsilon = c' \sqrt{\log mh}$, for large enough $c$ we have $2 \log n -c \frac{\epsilon^2}{K^4} \leq 2 \log n - c \log{mh}/K^4 \leq  -c''\frac{\log mh}{K^4}$ since $mh > n$ and that
\begin{equation}
\max_{i,j} |\widehat{S}_m(t_0, i, j) - \tilde{B}(t_0,i,j)|  \leq C \sqrt{\frac{\log mh}{mh}}
\end{equation}
with probability at least $1 - \frac{c}{m^4h^4}$.

\end{proof}

\subsection{Total error}

Putting the bias and variance together, we can bound the total error of the estimator. By the triangle inequality,
\begin{align*}
|\widehat{S}&_m(t_0, i, j) - B({t_0, i, j})| \le|\widehat{S}_m(t_0, i, j) - \tilde{B}(t_0,i,j)| + |\tilde{B}(t_0,i,j)-B({t_0, i, j})|.
\end{align*}
Hence,  
\begin{align*}
\max_{i,j} |\widehat{S}_m&(t_0, i, j) - B({t_0, i, j})| \\&= O_p \left(h+ \frac{1}{m^2 h^5} + \sqrt{\frac{\log m}{mh}}   \right).
\end{align*}
Optimizing over the order of $h$, 
we set $h \asymp  \left(\frac{\log m}{m}\right)^{1/3}$, giving
\begin{equation}\label{eq:bdd}
\max_{i,j} |\widehat{S}_m(t_0, i, j) - B({t_0, i, j})| \leq C  \left(\frac{\log m}{m}\right)^{1/3}.
\end{equation}
for some $C$. 

This completes the bound on the estimator error of $\widehat{S}_m(t_0)$. It remains to show that $\widehat{S}_m(t_0)$ is positive definite with high probability.


\subsection{Positive definiteness of $\widehat{S}_m(t_0)$ (Proof of Lemma \ref{lemma:varPSD2})}\label{app:posdef}
\begin{proof}

Let $u \in S^{n-1}$. Then 
\begin{align*}
u^T \hat S_m(t_0)  u &= \mathrm{vec}^T(uu^T) \mathrm{vec}(\hat S_m(t_0) ) \\&= \mathrm{vec}^T(uu^T)\left[\begin{array}{c} w^{(1,1)} , \dots ,w^{(n,n)}\end{array}\right]^T \mathrm{vec}\left(\hat{S} - \frac{\tr(A)}{m} I\right).
\end{align*}
Observe that $\|\mathrm{vec}^T(uu^T)\left[\begin{array}{c} w^{(1,1)} , \dots ,w^{(n,n)}\end{array}\right]^T\|_2 \leq c/\sqrt{mh}$, since the $w^{(ij)}$ have disjoint support. Thus by Lemma \ref{lemma:var}
\begin{equation}
\prob{u^T (\hat S_m(t_0)  - E[\hat S_m(t_0) ]) u > \epsilon \sqrt{\frac{c}{mh}}} \leq 2 \exp(-C\epsilon^2)
\end{equation}
Recall that $E[\hat S_m(t_0) ] = \tilde{B}(t_0)$.


Then by a standard argument using the union bound over an $\epsilon$ net of $S^{n-1}$, which has cardinality $\leq \exp (n \log(3/\epsilon))$,
\begin{align}
\label{Eq:HighProb}
&\prob{\exists  u \in S^{n-1} | u^T \left(S_m(t_0)  - \tilde{B}\right)u \leq c\epsilon  } \nonumber\\&\leq \exp (n \log(3/\epsilon)) 2 \exp(-C\epsilon^2 mh) \\\nonumber
&\leq C\exp\left(-c' \epsilon^2mh\right)
\end{align}
This holds for $c$ large enough, since $n < mh$.

Suppose that the event in \eqref{Eq:HighProb} holds. Then for all $u \in S^{n-1}$,
\begin{align*}
u^T S_m(t_0) u &\geq u^T \tilde{B}(t_0)u - c\epsilon\\\nonumber
&\geq u^T \tilde{B}(t_0) u (1-\delta)
\end{align*}
where $\delta \lambda_{min}(\tilde{B}(t_0)) \geq c\epsilon$. Note that since $\tilde{B}(t_0)$ is a positively-weighted average of matrices $B_i$ with minimum eigenvalues $\geq 1/c_b$ (assumption A4), $\lambda_{min}(\tilde{B}(t_0)) \geq 1/c_b$. By a similar argument, the upper bound holds. We thus have
\begin{equation}\label{eq:evnt}
(1 + \delta)\tilde{B}(t_0) \succeq \hat S_m(t_0) \succeq (1-\delta) \tilde{B}(t_0)
\end{equation}
with probability at least $1 - \frac{c'}{m^4h^4}$, for some fixed $\sqrt{\frac{c \log mh}{mh}}\leq \delta \leq 1$. Note that when \eqref{eq:evnt} holds, $S_m(t_0)$ is positive definite.  
\end{proof}

\subsection{Theorem \ref{thm}}
By the union bound the probability that the events $\mathcal{A}$ and $\mathcal{B}$ hold is $\mathbb{P}(\mathcal{A}\cap\mathcal{B}) = 1 - \frac{c}{m^4 h^4}$. 
Thus, combining the bound \eqref{eq:bdd} and the proof of positive definiteness in the previous subsection, the proof of Theorem \ref{thm} in the main text results.
\setcounter{theorem}{2}
\begin{theorem}
\label{thm}


Suppose that the conditions of Lemma \ref{lemma:varPSD} hold and $h \asymp \left(\frac{\log m}{m}\right)^{1/3}$. Then with probability at least $1 - \frac{c''}{m^{8/3}}$,
\[
\max_{ij} |\widehat{S}_m(t_0, i, j) - B({t_0, i, j})| \leq C \left(\frac{\log m}{m}\right)^{1/3}
\]
for some $C$, and $\widehat{S}_m(t_0, i, j)$ is positive semidefinite.
\end{theorem}

\section{Estimation error bound for $A$ part: Proof of Theorem \ref{lem:Apart2}}
\label{app:lem5}
\subsection{Trace Estimator}
We first bound the error for the estimator 
\begin{align}\label{eq:trdef}
\hat{\tr}(B) &= \sum_{i = 1}^m w_i \|X_t \|_2^2 - \frac{n}{m} \tr(A),\qquad
w_i = \frac{1}{m}  .
\end{align}
of the constant trace of $B$, ${\tr}(B)$.  

\begin{lemma}
\label{lem:trace}
Suppose that $\|A\| \leq c_A$ and $\|B(t)\|\leq c_B$ for all $t,m$, and $\tr(B(t))$ is constant over time. We have with probability $1-\frac{3}{m^5}$, 
\[
\frac{1}{n} |\hat{{\mathrm{tr}}}(B) -{ \mathrm{tr}}(B)| \le C(c_A + c_B)\sqrt{{\frac{\log m}{mn}}}.
\]
\end{lemma}

\begin{proof}
The bias of $\hat{\tr}(B(t_0))$ is zero since $\tr(B(t))$ is constant. 
To bound the variance, we can rewrite \eqref{eq:trdef} as
\begin{align*}
\hat{{\mathrm{tr}}}(B(t)) &= \|XW_{t}\|_F^2 - \frac{n}{m}  \mathrm{tr}(A),\\
W_t &= \mathrm{diag}(w).
\end{align*}

Note that 
\begin{align}
E \|XW_{t}\|_F^2  &= \sum_{i = 1}^m w_i\mathrm{tr}(B) + \frac{n}{m} \tr(A)\\
&= \tilde{\tr}({B}) + \frac{n}{m} \tr(A).
\end{align}
Also note that $\|XW_{t}\|_F^2 = \tr (\mathrm{vec}(X) (W_t \otimes I_n) \mathrm{vec}^T(X))$. Thus,
\begin{align*}
\frac{1}{n}\hat{{\mathrm{tr}}}(B(t)) &= \frac{1}{n} \tr \left(\mathrm{vec}(X) (W_t \otimes I_n) \mathrm{vec}^T(X)\right) -\frac{1}{m} \mathrm{tr}(A).
\end{align*}
Hence, by Lemma \ref{lemma:var}, with $\tilde{w} = \frac{1}{n} (w(t) \otimes 1_n)$ on the $mn$ indices corresponding to the diagonal elements of $\hat{S}=\mathrm{vec}(X) \mathrm{vec}(X)^T$, we have for $\frac{\epsilon}{\sqrt{mn}} = o(1)$,
\[
\prob{|\tilde{w}^T \mathrm{vec}(\hat{S} - \Sigma)| \geq \epsilon (c_A + c_B) \|\tilde{w}\|_2 } \leq 2\exp\left(-c\frac{\epsilon^2}{K^4}\right)
\]
and thus
\begin{align*}
&\prob{\left|\frac{1}{n}(\hat{{\mathrm{tr}}}(B) - {\tr}({B}))\right| \geq \epsilon (c_A + c_B) \frac{1}{\sqrt{mn}} } \leq 2\exp\left(-c\frac{\epsilon^2}{K^4}\right)
\end{align*}
since $\|w\|_2 \leq \frac{1}{\sqrt{mn}}$. Set $\epsilon = C \sqrt{\log m}$ with $C$ such that with probability at least $1 - \frac{3}{m^5}$,
\[
\frac{1}{n} |\hat{{\mathrm{tr}}}(B(t)) -{ \mathrm{tr}}({B}(t))| \le C(c_A + c_B)\sqrt{\frac{\log m}{mn}}.
\]
This concludes the proof.
\end{proof}

\subsection{Elementwise error (Proof of Theorem \ref{lem:Apart2})}
We can now prove the error bound for $\tilde{A}$.
\setcounter{theorem}{5}
\begin{theorem}
\label{lem:Apart2}
Suppose Assumptions [B2, A1] hold. 
Then 
\[
\max_{i,j} |\tilde{A}_{ij} - A_{ij}| \le C (c_A + c_B) \sqrt{\frac{\log m}{n}}
\]
with probability $1-\frac{c}{m^4}$ for some constants $C,c>0$.
\end{theorem}

\begin{proof}
Recall that 
\begin{align} 
\label{eq:atilde}
\tilde{A} &= \frac{1}{n} X^T X - \frac{1}{n}\mathrm{diag}\{\hat{\mathrm{tr}}(B(1/m)), \dots,\hat{\mathrm{tr}}(B(1))\}. 
\end{align}
Recall that $\hat{S} = \mathrm{vec}(X) \mathrm{vec}(X)^T$.
For $i\neq j$, by \eqref{eq:atilde} and the definition of Kronecker products we can then write
\begin{align*}
\tilde{A}_{ij} &= \frac{1}{n} X_{i}^T X_j\\
&= \frac{1}{n} \sum_{\ell = 1}^n \hat{S}_{\ell + (i-1) m, \ell + (j-1)m}
\end{align*}
where $X_i$ is the $i$th column of $X$. Thus, we can write
\[
\tilde{A}_{ij} = w^T \mathrm{vec}(\hat{S})
\]
for some $w \in \mathbb{R}^{m^2 n^2}$ with $n$ nonzero elements all equal to $1/n$. Clearly $\|w\|_2 = 1/\sqrt{n}$.
We can then apply Lemma \ref{lemma:var} with $w$ as the weight vector.
Using the union bound over $i,j$ and assuming $\|A\|,\|B(t)\|$ are bounded by $c_a, c_b$ respectively, Lemma \ref{lemma:var} thus gives
\begin{equation}
\max_{i \neq j} |\tilde{A}_{ij} - A_{ij}|  \leq C (c_A + c_B) \sqrt{\frac{\log m}{n}}
\end{equation}
with probability at least $1 - \frac{c}{m^4}$ for absolute constants $C,c$.

For the diagonal elements, Lemma \ref{lem:trace} shows that with probability $1-\frac{3}{m^5}$, 
\[
\frac{1}{n} |\hat{{\mathrm{tr}}}(B(t)) -{ \mathrm{tr}}(B(t))| \le C (c_A + c_B)\left(\frac{\log m}{mn}\right)^{1/2},
\]
and thus by Lemma \ref{lemma:var} and the union bound
\begin{align*}
\max_{i = j} |\tilde{A}_{ij} - A_{ij}|  &\leq C(c_A + c_B) \sqrt{\frac{\log m}{n}} + C_1(c_A + c_B) \left(\frac{\log m}{mn}\right)^{1/2} \\&\leq C (c_A + c_B) \sqrt{\frac{\log m}{n}}
\end{align*}
with probability at least $1-\frac{3}{m^4} - \frac{c}{m^4} > 1 - \frac{c}{m^4}$, since $m > n$.

Using the union bound gives
\begin{equation}
\max_{i , j} |\tilde{A}_{ij} - A_{ij}|  \leq C(c_A + c_B) \sqrt{\frac{\log m}{n}}
\end{equation}
with probability at least $1 - \frac{c}{m^4}$. 
\end{proof}

\section{Concentration bound (Lemma \ref{lemma:var})}
\label{sec:conc}
We use the following concentration bound to bound the error of the $A$ and $B(t)$ estimates. Note that it also gives the corresponding $A$ part and $B$ part bounds in \citep{rudelson2015high} as special cases.
\begin{lemma}[Concentration bound]
\label{lemma:var}
Let $w \in \mathbb{R}^{m^2n^2}$ and let $x = \Sigma^{1/2} z$ be a subgaussian random vector where $z_i$ are independent, zero mean, unit variance, and have $\|z_i\|_{\psi_2} \leq K$. Let $\hat{S} = xx^T$. Then if $\epsilon \|\Sigma\| \|w\|_2 = o(1)$
\begin{equation}
\mathbb{P}(|w^T \mathrm{vec}(\hat{S} - \Sigma)| \geq \epsilon \|\Sigma\| \|w\|_2  ) \leq 2 \exp\left( - c \frac{\epsilon^2}{K^4}\right)
\end{equation}
where $c$ is an absolute constant. 

\end{lemma}

\begin{proof}
By the definition of the vectorization operator and letting $\mathbf{x}_t(i) = [x_{t,(i-1)\bar{p}_1 +1},\dots,x_{t,i\bar{p}_1}]$,
\begin{align}
w^T\mathrm{vec}(\hat{S}) =  w^T\mathrm{vec}(x x^T) = x^T W x = z^T \Sigma^{1/2} W \Sigma^{1/2}z, 
\end{align}
where $W = \mathrm{vec}^{-1}(w)$ and $\mathrm{vec}^{-1}(\cdot)$ is the inverse of the vectorization operator, mapping $\mathbb{R}^{m^2n^2}$ to $\mathbb{R}^{mn \times mn}$.

Thus, by the Hanson-Wright inequality,
\begin{align*}
\mathbb{P}(|w^T \mathrm{vec}(\hat{S} - \Sigma)| \geq t  ) &\leq 2 \exp\left( - c \min\left(\frac{t^2}{K^4 \|\Sigma^{1/2} W \Sigma^{1/2}\|_F^2},\frac{t}{K^2 \|\Sigma^{1/2} W \Sigma^{1/2}\| }\right)\right)\\
&\leq 2 \exp\left( - c \min\left(\frac{t^2}{K^4 \|\Sigma\|^2 \|W\|_F^2}, \frac{t}{K^2 \|\Sigma\| \|W\| }\right)\right)\\
&\leq 2 \exp\left( - c \min\left(\frac{t^2}{K^4 \|\Sigma\|^2 \|w\|_2^2}, \frac{t}{K^2 \|\Sigma\| \|w\|_2}\right)\right)
\end{align*}
so
\begin{align*}
\mathbb{P}(|w^T \mathrm{vec}(\hat{S} - \Sigma)|\geq \epsilon \|\Sigma\| \|w\|_2  ) &\leq 2 \exp\left( - c \frac{\epsilon^2}{K^4}\right)
\end{align*}
where we set $t = \epsilon \|\Sigma\| \|w\|_2$ and assumed $\epsilon \|\Sigma\| \|w\|_2 = o(1)$.

\end{proof}

\section{Kernel Smoothing (Proof of Lemma \ref{lemma:bias1})}\label{app:kernel}

The following lemma bounds the bias inherent in using a kernel to smooth the sample covariance matrix for $B(t)$ across time.

\begin{proof}[\textbf{Proof of Lemma \ref{lemma:bias1}}]
The proof is found in Lemma 5 of \citep{zhou:TV}, repeated here for completeness.

Without loss of generality, let $t = t_0$. We will use the Riemann integral to approximate the sum
\begin{align*}
&\tilde{B}_p(t_0, i,j)  = \frac{1}{m} \sum_{k=1}^m  \frac{2}{h} K \left(\frac{\frac{k}{m} -t_0}{h} \right)B_{i,j} \left(\frac{k}{m}\right).
\end{align*}

Specifically, 
\begin{align*}
\tilde{B}_p(t_0, i,j) =& \int_{0}^{1} \frac{2}{h} K\left(\frac{u-t_0}{h}\right) {B}_{i,j}(u) du  + O\left(\frac{1}{m^2  h^5}  \right)
\end{align*}
where the second equality follows from Assumption K5 and the assumed bound on the second derivative of $B_{ij}(t)$.
From the proof of Lemma 5 in \citep{zhou:TV}, $\int_{0}^{1} \frac{2}{h} K(\frac{u-t_0}{h}) {B}_{i,j}(u) du - B_{i,j}(t_0) = O(h)$ so 
\begin{align*}
\tilde{B}_p(t_0, i,j) - B_{i,j}(t_0) &= O\left(h  + \frac{1}{m^2  h^5}\right).
\end{align*}
Taking the maximum over $i,j$ and $t_0 \in [0,1]$ completes the proof.

\end{proof}

\end{document}